\documentclass[twoside]{article}
\usepackage{fullpage}
\usepackage[utf8]{inputenc}
\usepackage{amsmath,amssymb,graphicx}
\usepackage{url}
\usepackage{cite}
\usepackage{csquotes}

\usepackage{timetravel}
\setCounters{theorem,lemma,casecounter,subcasecounter,subsubcasecounter}

\usepackage{subcaption}

\usepackage{amsthm}
\theoremstyle{plain}
\newtheorem{theorem}{Theorem}
\newtheorem{lemma}[theorem]{Lemma}
\newtheorem{corollary}[theorem]{Corollary}
\newtheorem{property}[theorem]{Property}

\theoremstyle{definition}
\newtheorem{question}{Question} 

\usepackage{hyperref}
\usepackage{xcolor}

\usepackage[inline]{enumitem}
\usepackage{hyperref} 

\usepackage{float}

\newcounter{casecounter}
\newcounter{subcasecounter}
\newcounter{subsubcasecounter}
\makeatletter
\newcommand{\ccase}[1]{\stepcounter{casecounter}\setcounter{subcasecounter}{0}\protected@write \@auxout {}{\string \newlabel {#1}{{\thecasecounter}{\thepage}{\thecasecounter}{#1}{}} }\hypertarget{#1}{\noindent\textbf{Case \thecasecounter.}}
}

\newcommand{\subcase}[1]{\stepcounter{subcasecounter}\setcounter{subsubcasecounter}{0}\protected@write \@auxout {}{\string \newlabel {#1}{{\thecasecounter.\thesubcasecounter}{\thepage}{\thecasecounter.\thesubcasecounter}{#1}{}} }\hypertarget{#1}{\noindent\textbf{Case \thecasecounter.\thesubcasecounter.}}
}

\newcommand{\subsubcase}[1]{\stepcounter{subsubcasecounter}\protected@write \@auxout {}{\string \newlabel {#1}{{\thecasecounter.\thesubcasecounter.\thesubsubcasecounter}{\thepage}{\thecasecounter.\thesubcasecounter.\thesubsubcasecounter}{#1}{}} }\hypertarget{#1}{\noindent\textbf{Case \thecasecounter.\thesubcasecounter.\thesubsubcasecounter.}}
}
\makeatother

\newcommand{\R}{{{\sf l} \kern -.10em {\sf R} }}
\newcommand{\eps}{\varepsilon}
\newcommand{\wlogc}{W.l.o.g.}
\newcommand{\wlogs}{w.l.o.g.}

\graphicspath{{./figures/}}

\usepackage{fancyhdr}
\title{Lombardi Drawings of Knots and Links}

\author{
	Philipp Kindermann\thanks{
  Universit\"at W\"urzburg, Germany, 
	\url{philipp.kindermann@uni-wuerzburg.de}}
\and
	Stephen Kobourov\thanks{
  University of Arizona, Tucson, AZ, US, 
  \url{kobourov@cs.arizona.edu}  }
\and
	Maarten L\"offler\thanks{
  Universiteit Utrecht, the Netherlands, 
  \url{m.loffler@uu.nl}}
\and
	Martin~N\"ollenburg\thanks{
  TU Wien, Vienna, Austria, 
  \url{noellenburg@ac.tuwien.ac.at}}
\and
	Andr\'e Schulz\thanks{
  FernUniversit\"at in Hagen, Germany, 
  \url{andre.schulz@fernuni-hagen.de}}
\and
	Birgit Vogtenhuber\thanks{
  Graz University of Technology, Austria,  
  \url{bvogt@ist.tugraz.at}}
}

\begin{document}

\maketitle
\setcounter{footnote}{0}

\begin{abstract}
Knot and link diagrams are projections of one or more 3-dimensional simple closed curves into  
$\R^2$, such that no more than two points 
project to the same point in $\R^2$. 
These diagrams are 
drawings of 4-regular plane multigraphs.
Knots are typically smooth curves in $\R^3$, so their projections should be smooth curves in $\R^2$ with good continuity and large crossing angles: exactly the properties of Lombardi graph drawings (defined by circular-arc edges and perfect angular resolution).

We show that several knots do not allow plane Lombardi drawings. On the other hand, we identify a large class of 4-regular plane multigraphs that do have Lombardi drawings.
We then study two relaxations of Lombardi drawings and 
show that every knot 
admits a plane 2-Lombardi drawing (where edges are composed of two circular arcs). 
Further, every knot 
is \emph{near-Lombardi}, that is, it can be drawn as Lombardi drawing when 
relaxing the angular resolution requirement by an arbitrary small angular offset~$\eps$, while maintaining a $180^\circ$ angle between opposite edges.
\end{abstract}

\bigskip
\bigskip
\begin{figure}[H]
		\centering
		\subcaptionbox{\label{sfg:teaser-rolfsen}}{\includegraphics[scale=.75]{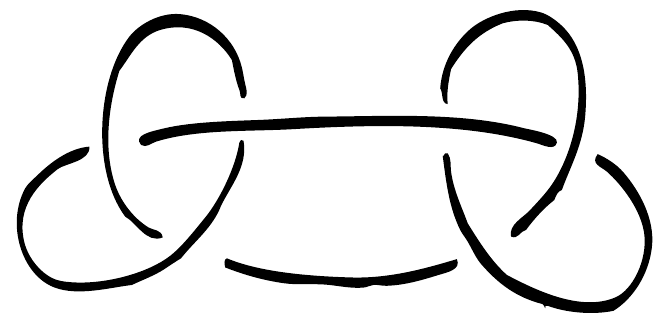}}
		\hfill
		\subcaptionbox{\label{sfg:teaser-livingston}}{\includegraphics[scale=.75]{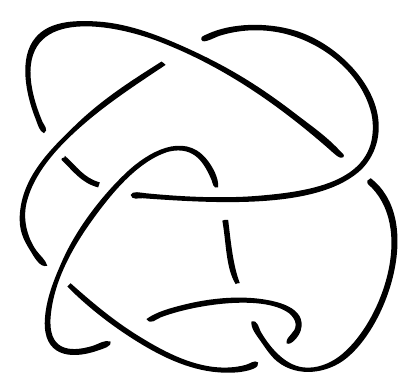}}
		\hfill
		\subcaptionbox{\label{sfg:teaser-kauffman}}{\includegraphics[scale=.12, angle=-2]{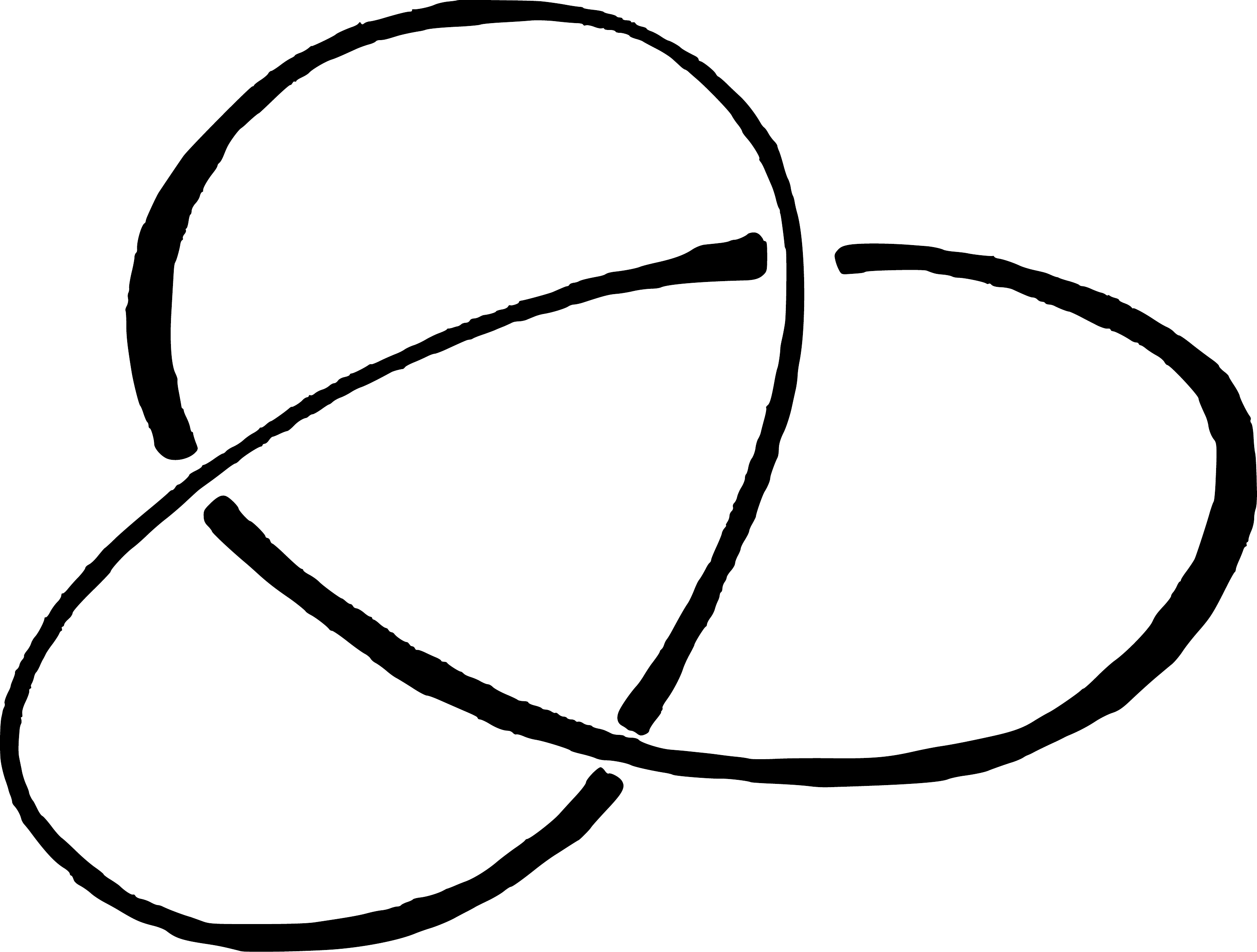}}
		\caption{Hand-made drawings of knots from the books of (a) Rolfsen~\cite{rolfsen1976knots}, 
		(b) Livingston~\cite{livingston}, and (c) Kauffman~\cite{kauffman}.}
		\label{fig:teaser}	
\end{figure}

\clearpage

\section{Introduction}
A \emph{knot} is an embedding of a simple closed curve in 3-dimensional Euclidean space~$\R^3$.
Similarly, a \emph{link} is an embedding of a collection of simple closed curves in $\R^3$.
A \emph{drawing of a knot (link)} (also known as \emph{knot diagram}) is a projection of the knot (link) to the Euclidean plane~$\R^2$ such that for any point of $\R^2$, at most two points of the curve(s) are mapped to it~\cite{SchareinPhD,rolfsen1976knots,cromwell2004knots}. 
From a graph drawing perspective, drawings of knots and links are drawings of 4-regular plane multigraphs that contain neither loops nor cut vertices. 
Likewise, every 4-regular plane multigraph without loops and cut vertices can be interpreted as a link.
Unless specified otherwise, we assume that a multigraph has no self-loops or cut vertices.

In this paper, we address a question that was recently posed by Benjamin Burton: 
\enquote{Given a drawing of a knot, how can it be redrawn \emph{nicely} without changing the given topology of the drawing?} 
We do know what a drawing of a knot is, but what is meant by a \emph{nice} drawing? 
Several graphical annotations of knots and links as graphs have been proposed in the knot theory literature, but most of the illustrations are hand-drawn; see Fig.~\ref{fig:teaser}. 
When studying these drawings, a few desirable features become apparent: 
\begin{enumerate*}[label=(\roman*)]
	\item edges are typically drawn as smooth curves, 
	\item the angular resolution of the underlying 4-regular graph is close to $90^\circ$, and 
	\item the drawing preserves the continuity of the knot, that is, in every vertex of the underlying 
graph, opposite edges have a common tangent. 
\end{enumerate*}
There are many more features one could wish from a drawing of a knot or link, see, e.g., the energy models discussed in the PhD thesis of Scharein~\cite{SchareinPhD}.
But our task is to \emph{redraw} a given drawing of a knot with a particular topology, so other typical quality metrics, such as the number of crossings, that vary with the choice of the embedding or topology of a knot diagram do not apply here.

There already exists a graph drawing style that fulfills the three requirements above: 
a \emph{Lombardi drawing} of a (multi-)graph $G=(V,E)$ is a drawing of $G$ in the Euclidean plane with the following properties:
\begin{enumerate}[nolistsep]
	\item The vertices are represented as distinct points in the plane
	\item The edges are represented as circular arcs connecting the representations of their end vertices (and not containing the representation of any other vertex); note that a straight-line segment is a circular arc with radius infinity.
	\item Every vertex has \emph{perfect angular resolution}, i.e., its incident edges are equiangularly spaced. For knots and links this means that the angle between any two consecutive edges is $90^\circ$.
\end{enumerate}
A Lombardi drawing is plane if none of its edges intersect. We are particularly interested in plane Lombardi drawings, since crossings change the topology of the drawn knot.

\paragraph{Knot diagram representations.}
There are several ways in the literature to combinatorically represent a knot 
diagram that are different from the 4-regular multi-graph as described above.which we will briefly survey.
The Alexander-Briggs-Rolfsen notation~\cite{ab-otkc-AN27,rolfsen1976knots} is a 
well established notation that organizes knots by their vertex number 
and a counting index, e.g., the trefoil knot  $3_1$ is listed as the first 
(and only) knot with three vertices. The Gauss code~\cite{fm-cgc-dcg99} of a knot 
can be computed as follows. Label each
vertex with a letter, then pick a starting vertex and a direction, traverse the
knot, and record the labels of the vertices encountered in the order of the traversal
with a preceding ``$-$'' if the part of the knot that is followed at the vertex lies
below the other part (called an \emph{under crossing}); see Fig.~\ref{sfg:notation-gauss}.
The Dowker–Thistlethwaite code~\cite{dt-ckp-ta83} is obtained similar to the Gauss code: 
Pick a starting vertex and a direction, traverse the knot, and label the
vertices in the order of the traversal with consecutive integers, starting
from 1, with a preceding ``$-$'' in case of an under crossing for even labels. 
Then, every vertex has two labels: a positive odd label and an even label. Order 
the vertices ascendingly by their odd label, and record their corresponding
even labels in this order; see Fig.~\ref{sfg:notation-dowker}.

\begin{figure}[t]
		\centering		
		\begin{subfigure}[b]{.47\textwidth}
			\centering
			\includegraphics[scale=1,page=1]{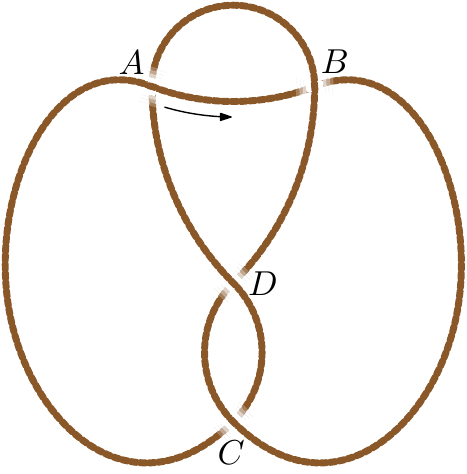}
			\caption{Gauss code: $A,-B,C,-D,B,-A,D,-C$}
      \label{sfg:notation-gauss}
		\end{subfigure}
		\hfil
		\begin{subfigure}[b]{.5\textwidth}
			\centering
			\includegraphics[scale=1,page=2]{notations}
			\caption{Dowker-Thistlethwaite code: $-6,-8,-2,-4$}
      \label{sfg:notation-dowker}
		\end{subfigure}		
		\caption{Representations of the knot $4_1$.}
		\label{fig:notation}	
\end{figure}

\paragraph{Knot drawing software.}
Software for generating drawings for knots and links exists. One powerful package is {\tt KnotPlot}~\cite{SchareinPhD}, 
which provides several methods for drawing knot diagrams. 
It contains a library of over 1,000 precomputed knots and can also generate knot drawings of certain families, such as torus
knots.  
{\tt KnotPlot} is mainly concerned with visualizing knots in three and four dimensions. To this end, a knot is represented as a 3-dimensional path on a set number of nodes, and then forces are used on these nodes to smoothen the visualization without changing the topology.
But {\tt KnotPlot} also provides methods for drawing general knots in 2D based on the embedding of the underlying plane multigraph, represented by the Dowker-Thistlethwaite code.
By replacing every vertex by a 4-cycle, the multigraph becomes a simple planar 3-connected graph, which is then drawn using Tutte's barycentric method~\cite{T63}. In the end, the modifications are reversed and a drawing of the knot is obtained with
edges drawn as polygonal arcs. The author noticed that this method \enquote{... does not yield
\enquote{pleasing} graphs or knot diagrams.} In particular, he noticed issues 
with vertex and angular resolution~\cite[pg. 102]{SchareinPhD}.

Another approach was used by Emily Redelmeier~\cite{redelmeier}
in the Mathematica package \texttt{KnotTheory}. Here, every arc, crossing, and face of the
knot diagram is associated with a disk. The drawing is then generated from the implied circle packing
as a circular arc drawing. As a result of the construction, every edge in the diagram is made of three circular arcs with common tangents at opposite edges. Since no further details are given, it is hard to evaluate the effectiveness of this approach, although as we show in this paper, three circular arcs per edge are never needed.
A related drawing style for knots are the so-called \emph{arc presentations}~\cite{cromwell1998arc}. An arc presentation is an
orthogonal drawing, that is, all edges are sequences of horizontal and vertical segments, with the additional properties
that at each vertex the vertical segments are above the horizontal segments in the corresponding knot and that each
row and column contains exactly one horizontal and vertical segment, respectively. However, these drawings might require a 
large number of bends per edge.

\paragraph{Lombardi drawings.}
Lombardi drawings were introduced by Duncan et al.~\cite{Lombardi1}. They showed that 
2-degenerate graphs have Lombardi drawings and that 
all $d$-regular graphs, with $d \not\equiv 2 \pmod{4}$, have Lombardi drawings with all vertices placed along a common circle. 
Neither of these results, however, is guaranteed to result in plane drawings. Duncan et al.~\cite{Lombardi1} also showed that there exist planar graphs that do not have plane Lombardi drawings, but restricted graph 
classes (e.g., Halin graphs) do. In subsequent work, Eppstein~\cite{Lombardi2_GD,e-mpdasbpld-14} showed that every (simple) planar 
graph with maximum degree three has a plane Lombardi drawing. Further, he showed that a certain class of 4-regular 
planar graphs (the medial graphs of polyhedral graphs) also admit plane Lombardi drawings and he presented an example 
of a 4-regular planar graph that does not have a plane Lombardi drawing. A generalization of Lombardi drawings are
\emph{$k$-Lombardi drawings}. Here, every edge is a sequence of at most~$k$ circular arcs that meet at a common tangent. 
Duncan et al.~\cite{degkl-ppld-12} showed that every planar graph has a plane $3$-Lombardi drawing. 
Related to $k$-Lombardi-drawings are \emph{smooth-orthogonal drawings of complexity $k$}~\cite{Smoothorthogonalintro}. 
These are plane drawings where every edge consists of a sequence of at most $k$ quarter-circles and axis-aligned segments 
that meet smoothly, edges are axis-aligned (emanate from a vertex either horizontally or vertically),  and no two edges emanate in the same direction. Note that in the special case of 4-regular graphs, smooth-orthogonal drawings of 
complexity $k$ are also plane $k$-Lombardi drawings.

\paragraph{Our Contributions.}
The main question we study here is motivated by the application of the Lombardi drawing style to knot and link drawings: Given a 4-regular plane multigraph~$G$ without loops and cut vertices, 
does $G$ admit a plane Lombardi drawing with the same combinatorial embedding? 
In Section~\ref{sec:general} we start with some positive results  on extending a plane Lombardi drawing, as well as composing two plane Lombardi drawings.
In Section~\ref{sec:circlepacking}, by extending the results of Eppstein~\cite{Lombardi2_GD,e-mpdasbpld-14}, we show that a large class of multigraphs, including 4-regular polyhedral graphs, does have plane Lombardi drawings.
Unfortunately, there exist several small knots that do not have a plane Lombardi drawing.
Section~\ref{sec:negative} discusses these cases but also lists a few positive results for small examples. 
In Section~\ref{sec:2lombardi}, we show that every 4-regular plane multigraph has a plane 2-Lombardi drawing. 
In Section~\ref{sec:almost}, we show that every 4-regular plane multigraph can be drawn
with non-crossing circular arcs, so that the perfect angular 
resolution criterion is violated only by an arbitrarily small value~$\eps$, while maintaining that opposite edges have common tangents.

\section{General Observations}\label{sec:general}

If a knot or a link has an embedding with minimum number of vertices that admits a plane Lombardi drawing, we call it a \emph{plane Lombardi knot} (\emph{link}). 
We further call the property of admitting a plane Lombardi drawing \emph{plane Lombardiness}.
If two vertices in a plane Lombardi drawing of a knot are connected by a pair of multi-edges, we denote the face enclosed by these two edges as a \emph{lens}.

There exist a number of operations that maintain the plane Lombardiness of a 4-regular multigraph.
Two knots $A$ and~$B$ can be combined by \emph{connecting $A$ and $B$ along edges $a$ of $A$ and $b$ of $B$},
that is, cutting an edge $a$ of $A$ and an edge $b$ of $B$ open and gluing pairwise the loose ends of of~$a$ with the loose ends of~$b$.
This operation is known as a \emph{knot sum} $A+B$.
Knots that cannot be decomposed into a sum of two smaller knots are known as 
\emph{prime knots}. By Schubert's theorem, every knot can be uniquely decomposed 
into prime knots~\cite{s-sitz-49}. The smallest prime knot 
is the trefoil knot with three crossings or vertices; see Fig.~\ref{sfg:teaser-kauffman}. 
Rolfsen's knot table\footnote{\url{http://katlas.org/wiki/The_Rolfsen_Knot_Table}} 
lists all prime knots with up to ten vertices. 

\begin{theorem}\label{thm:summation}
	Let $A$ and $B$ be two 4-regular multigraphs with plane Lombardi drawings. 
	Let $a$ be an edge of $A$ and $b$ an edge of $B$. 
	Then the knot sum  $A + B$, obtained by connecting $A$ and $B$ along edges $a$ and $b$, admits a plane Lombardi drawing.
\end{theorem}

\begin{proof}
  We first apply a Möbius transformation to the plane Lombardi drawings of $A$ 
  and~$B$ so that in the resulting drawings the given edges $a$ and $b$ are 
  drawn as straight edges passing through the point at infinity, i.e., they are 
  complements of line segments on an infinite-radius circle; see Fig.~\ref{fig:summation}. 
  Next, we rotate and align both of these drawings so that 
  edges~$a$ and~$b$ are collinear and the subdrawings obtained by removing 
  edges $a$ and $b$ do not intersect. In the final step, we remove both $a$ 
  and $b$ and reconnect their vertices by two new edges $c$ and $d$ connecting the 
  two drawings, one being a line segment and the other passing through 
  infinity. Since Möbius transformations preserve planarity and Lombardiness 
  and our construction does not introduce any edge crossings, the resulting 
  drawing is a plane Lombardi drawing. Another Möbius transformation may be 
  applied to remove the edge through infinity.
\end{proof}

Another operation that preserves the plane Lombardiness is \emph{lens multiplication}.
Let $G=(V,E)$ be a 4-regular plane multigraph with a lens between two vertices~$u$ and~$v$. 
A lens multiplication of~$G$ is a 4-regular plane multigraph that is obtained 
by replacing the lens between~$u$ and~$v$ with a chain of lenses.

\begin{figure}[t]
	\centering
  \begin{minipage}[b]{.39\textwidth}
    \centering
    \includegraphics[scale=1,page=2]{lenses}
    \caption{Adding two plane Lombardi drawings of 4-regular multigraphs.}	
    \label{fig:summation}
  \end{minipage}
	\hfill
  \begin{minipage}[b]{.27\textwidth}
		\centering
			\includegraphics[scale=1,page=1]{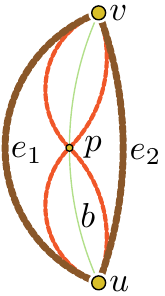}
		\caption{Subdividing a lens between~$u$ and~$v$ by a new vertex~$p$.}
		\label{fig:lens-subdivision}
  \end{minipage}
  \hfill
  \begin{minipage}[b]{.29\textwidth}
    \centering
      \includegraphics[scale=1]{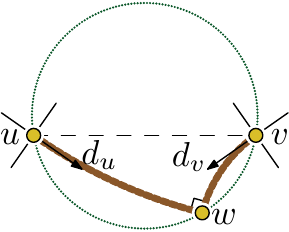}
    \caption{Placement circle for neighbor $w$ of $u$ and $v$ in a 4-regular graph.}
    \label{fig:placement}
  \end{minipage}
\end{figure}

\begin{lemma}\label{lem:lensmult}
	Let $G=(V,E)$ be a 4-regular plane multigraph with a plane Lombardi drawing~$\Gamma$. 
	Then, any lens multiplication $G'$ of $G$ also admits a plane Lombardi drawing.
\end{lemma}

\begin{proof}
	Let $f$ be a lens in $\Gamma$ spanned by two vertices $u$ and $v$. 
	We denote the two edges bounding the lens as $e_1$ and $e_2$. 
	The angle between $e_1$ and $e_2$ in both end-vertices is $90^\circ$. 
	We define the bisecting circular arc $b$ of $f$ as the unique circular arc 
	connecting $u$ and $v$ with an angle of $45^\circ$ to both $e_1$ and $e_2$; 
	see Fig.~\ref{fig:lens-subdivision}.
	
	Let $p$ be the midpoint of $b$.
	If we draw circular arcs from both $u$ and $v$ to~$p$ that have the same tangents as $e_1$ and $e_2$ 
	in $u$ and $v$, then these four arcs meet at $p$ forming angles of $90^\circ$. 
	Furthermore, each such arc lies inside lens $f$ and hence does not cross any other arc of~$\Gamma$. 
	The resulting drawing is thus a plane Lombardi drawing of a 4-regular multigraph 
	that is derived from $G$ by subdividing the lens~$f$ with a new degree-4 vertex.
	
	By repeating this construction inside the new lenses, we can create plane Lombardi 
	drawings that replace lenses by chains of smaller lenses.
\end{proof}

We will use the following property several times throughout the paper.

\begin{property}[Property~2 in~\cite{Lombardi1,degkl-ppld-12}]\label{prop:placement}
	Let $u$ and $v$ be two vertices with given positions that have a common, unplaced neighbor $w$. 
	Let $d_u$ and $d_v$ be two tangent directions and let  $\theta$ be a target angle. 
	Let $C$ be the locus of all positions for placing $w$ so that 
	(i) the edge $(u,w)$ is a circular arc leaving $u$ in direction $d_u$, 
	(ii) the edge $(v,w)$ is a circular arc leaving~$v$ in direction $d_v$, and 
	(iii) the angle formed at $w$ is $\theta$. Then $C$ is a circle, the so-called \emph{placement circle} of $w$.
\end{property}

Duncan et al.~\cite{degkl-ppld-12} further specify the radius and center of the placement circle by the input coordinates and angles. 
For the special case that the two tangent directions $d_u$ and $d_v$ are symmetric with respect to the line through $u$ and $v$, and that the angle $\theta$ is $90^\circ$ or $270^\circ$, the corresponding placement circle is such that its tangent lines at $u$ and $v$ form an angle of $45^\circ$ with the arc directions~$d_u$ and~$d_v$. 
In particular, the placement circle bisects the right angle between $d_u$ (resp.~$d_v$) and its neighboring arc direction. 
Fig.~\ref{fig:placement} illustrates this situation.

\section{Plane Lombardi Drawings via Circle Packing}
\label{sec:circlepacking}

Recall that \emph{polyhedral graphs} are simple planar 3-connected graphs, and that those graphs have a unique (plane) combinatorial embedding.
The (plane) \emph{dual graph}~$M'$ of a plane graph~$M$ has a vertex for every face of~$M$ and an edge between two vertices for every edge shared by the corresponding faces in $M$.
In the \enquote{classic} drawing $D(M,M')$ of a primal-dual graph pair $(M,M')$, every vertex of~$M'$ lies in its corresponding face of $M$ and vice versa, and every edge of~$M'$ intersects exactly its corresponding edge of $M$.
Hence, every cell of $D(M, M')$ has 
exactly two such edge crossings and exactly one vertex of each of~$M$ and~$M'$ on its boundary.
The \emph{medial graph} of a primal-dual graph pair $(M,M')$ has a vertex for every crossing edge pair in $D(M, M')$ and an edge between two vertices whenever they share a cell in $D(M, M')$; see Fig.~\ref{sfg:ext_1}. 
Every cell of the medial graph contains either a vertex of~$M$ or a vertex of $M'$ and every edge in the medial graph is incident to exactly one cell in $D(M,M')$.

Every 4-regular plane multigraph $G$ can be interpreted as the medial graph of some plane graph $M$ and its dual $M'$, where both graphs possibly contain multi-edges.
In fact, medial graphs have already been used in the context of knot diagrams by Tait in 1879~\cite{tait-knots}.
If~$G$ contains no loops and cut vertices, then neither $M$ nor $M'$ contains loops.
Eppstein~\cite{Lombardi2_GD} showed that if $M$ (and hence also $M'$) is polyhedral, then $G$ admits a plane Lombardi drawing.
We show next how to extend this result to a larger graph class.
We next show that if one of $M$ and $M'$ is simple, then $D(G)$ admits a plane Lombardi-drawing.
A full construction example of the algorithm can be found in Appendix~\ref{app:circle_packings}.

\begin{theorem}\label{thm:polyhedral}
    Let $G=(V,E)$ be a biconnected 4-regular plane multigraph and let~$M$ and~$M'$ be the primal-dual multigraph pair for which~$G$ is the medial graph. 
	If one of $M$ and $M'$ is simple, then $G$ admits a plane Lombardi drawing preserving its embedding.
\end{theorem}

\begin{proof}
	Assume \wlogs\ that $M$ is simple. 
	If $M$ (and hence also $M'$) is polyhedral, then $G$ admits a plane Lombardi drawing $\Gamma$ by Eppstein~\cite{Lombardi2_GD}.
	
	It remains to show that $\Gamma$ preserves the embedding of~$G$. The drawing~$\Gamma$ is constructed in the following way:
	Consider a primal-dual circle packing $C(M,M')$ of $(M,M')$ which exists due to Brightwell and Schreinerman~\cite{brightwell:1993}. 
	The plane Lombardi drawing $\Gamma$ of $G$ then is essentially the Voronoi diagram of $C(M,M')$.
As the combinatorial embedding of $M$ and $M'$ is unique up to homeomorphism on the sphere, there exists a M\"obius transformation $\tau$ such that the circle packing $\tau(C(M,M'))$ has the same unbounded face as $D(M,M')$.
	Hence, $\Gamma$ is a plane Lombardi drawing of~$G$ that preserves its combinatorial embedding. 

	\begin{figure}[t]
		\centering
		\hfill
		\begin{subfigure}[b]{.4\textwidth}
			\centering
			\includegraphics[scale=1, page=2]{extend_jocg}
			\caption{}\label{sfg:ext_1}
		\end{subfigure}
		\hfill
		\begin{subfigure}[b]{.4\textwidth}
			\centering
			\includegraphics[scale=1, page=4]{extend_jocg}
			\caption{}\label{sfg:ext_2}
		\end{subfigure}	
		\hfill \hfill \hfill \\[2ex]
		\hfill
		\begin{subfigure}[b]{.4\textwidth}
			\centering
			\includegraphics[scale=1, page=14]{extend_jocg}
\caption{}\label{sfg:extension_br}
		\end{subfigure}
		\hfill
		\begin{subfigure}[b]{.4\textwidth}
			\centering
			\includegraphics[scale=0.8, page=1]{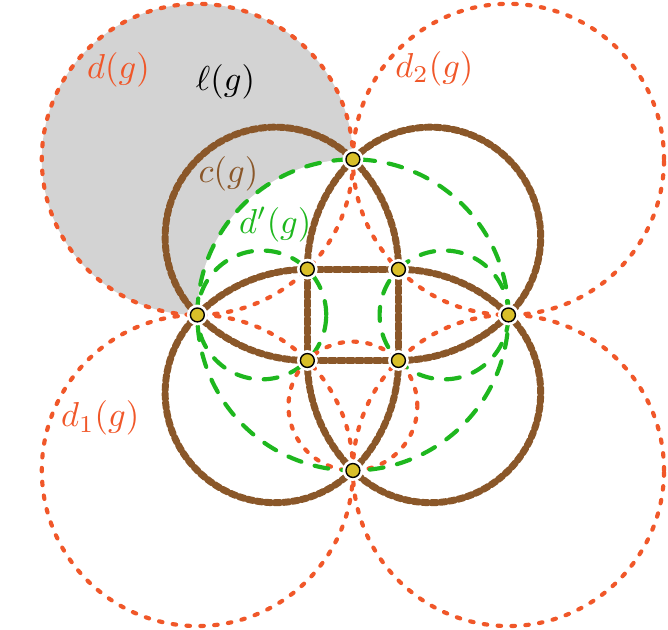}
			\caption{}\label{sfg:full_packing}
		\end{subfigure}				
		\hfill \hfill 
		\caption{(a)--(b) Modifications due to the addition of edge $e$.
			(c)~Extension of (b) to a polyhedral graph, and 
			(d)~the according primal-dual circle packing representation. 
			The medial graph~$G$ is drawn solid, 
			the primal multigraph~$M$ is drawn dotted, 
			and the dual multigraph~$M'$ is drawn dashed.
			The shaded area is the lens region~$l(g)$. \vspace{-0.4cm} }
		\label{fig:ext_full_packing}
	\end{figure}

	Now assume that $M$ is not 3-connected. 
	As a first step, we iteratively extend $M=M_0$ by adding $p$ edges until we obtain a polyhedral graph $M_p$.
	During this process, we also iteratively adapt the dual graph and the medial graph; 
see Figs.~\ref{sfg:ext_1}--\subref{sfg:ext_2} for an illustration.
Let~$M_{i+1}$ be the graph obtained from $M_{i}$ by adding edge $e$ to $M_i$. 
	The edge $e$ splits a face~$f$ of~$M_i$ with at least four incident vertices into two faces $f_1$ and $f_2$ with at least three incident vertices each. 
	In $M'_i$, the according vertex $f'$ is split into two vertices $f'_1$ and $f'_2$. 
	The edges incident to $f'$ are partitioned into edges incident to $f'_1$ and $f'_2$
	and an additional edge between~$f'_1$ and~$f'_2$ is added.
	In $G_i$, the edges inside the face~$f$ of $M_i$ form a cycle that connects every pair of edges in $M_i$ that is incident along the boundary of~$f$. 
	When $e$ is added, exactly two edges $g_1$, $g_2$ of $G_i$ are intersected by $e$. 
	To obtain~$G_{i+1}$, the edges $g_1$ and $g_2$ are replaced by four new edges, where each new edge has the new crossing between~$e$ and $(f'_1,f'_2)$ as one endpoint and one of the four endpoints of~$g_1$ and~$g_2$, respectively, as the other endpoint. 

	In the second step, we apply the result of Eppstein~\cite{Lombardi2_GD} to obtain a plane Lombardi drawing $\Gamma_p$ of $G_p$ together with a primal-dual circle packing $C(M_p,M'_p)$.
	Before going into the third step, the iterative removal of the edges that were added in the first step, let us consider the structure obtained from the second step in more detail; see Figs.~\ref{sfg:extension_br}--\ref{sfg:full_packing}.
For an edge $g$ of~$G_p$, consider the unique vertex $m(g) \in M_p$ that lies in a cell of~$G_p$ incident to~$g$. 
	Note that $g$ has its endpoints on two edges incident to $m(g)$ and adjacent in their order around $m(g)$. Let these edges be $(m(g),m_1(g))$ and $(m(g),m_2(g))$, respectively.
	Let~$d(g)$,  $d_1(g)$, and $d_2(g)$ be the disks in $C(M_p)$ corresponding to $m(g)$, $m_1(g)$, and $m_2(g)$, respectively.
	Then in~$\Gamma_p$, the circular arc $c(g)$ corresponding to $g$ lies in the interior\footnote{Here, \emph{interior} is meant w.r.t.\ the circle packing. 
		Note that a circle could also be inverted, that is, contain the unbounded face.}
	of the disk~$d(g)$ and has its endpoints on the touching points of $d(g)$ with $d_1(g)$ and~$d_2(g)$, respectively.
	These touching points are consecutive along the boundary of $d(g)$.
	Further, there is a disk $d'(g)$ in $C(M'_p)$ whose boundary intersects the boundary of $d(g)$ exactly in the endpoints of $c(g)$. 
	The intersection of $d(g)$ and $d'(g)$ contains $c(g)$ in its interior. 
	The circles~$\partial d(g)$ and~$\partial d'(g)$ intersect with right angles and $c(g)$ bisects the angles at both intersections.  
	We call $d(g) \cap d'(g)$ the \emph{lens region} $\ell(g)$ of $g$. 
	For any two edges $g_1$ and $g_2$ of~$\Gamma_p$, the according lens regions~$\ell(g_1)$ and~$\ell(g_2)$ are interior-disjoint.
	The lens regions of the edges incident to the face in $\Gamma_p$ corresponding to $m(g)$ cover the whole boundary of $d(g)$ and the endpoints of those regions appear in the same cyclic order as the according edges in $M_p$.

	In the third step, we iteratively remove the edges that were added in the first step, by this constructing a sequence of plane Lombardi drawings $\Gamma_i$ for~$G_i$, for $i=p-1,\ldots,0$. For any edge $g$ of $G_i$, consider the unique vertex $m(g) \in M_i$ that lies in a cell of $\Gamma_i$ incident to $g$, with endpoints on edges $(m(g),m_1(g))$ and $(m(g),m_2(g))$ of $M_i$, respectively.
	Let $d(g)$,  $d_1(g)$, and $d_2(g)$ be the disks in~$C(M_p)$ corresponding to $m(g)$, $m_1(g)$, and $m_2(g)$, respectively, and let $c(g)$ be the circular arc in $\Gamma_i$ corresponding to~$g$. 
We keep the following invariants for all edges $g$ of the drawing $\Gamma_i$: \begin{enumerate}[nosep,label=(\roman*)]
    \item $c(g)$ lies in the disk~$d(g)$ and has its endpoints on the touching points of $d(g)$ with $d_1(g)$ and $d_2(g)$, respectively.
	\item There is a disk $d'(g)$ whose boundary intersects the boundary of~$d(g)$ exactly in  $d(g) \cap d_1(g)$ and $d(g) \cap d_2(g)$, such that
		$c(g)$ bisects one of the two regions $d(g) \cap d'(g)$ and $d(g) \cap \mathbb{R}^2\setminus d'(g)$, which we call its lens region~$\ell(g)$.
	\item For any two edges $g_1$ and $g_2$ of $G_i$, the lens regions $\ell(g_1)$ and~$\ell(g_2)$ are interior-disjoint.
    \item The lens regions of the edges incident to the face in~$D(G_i)$ corresponding to $m(g)$ cover the whole boundary of $d(g)$ and the endpoints of those regions appear in the same cyclic order as the according edges in~$D(M_i)$.
  \end{enumerate}

	Obviously, those invariants are fulfilled by $\Gamma_p$. Hence, assume that they are also fulfilled for $\Gamma_{i+1}$, and consider the removal of the edge $e=(v_1,v_2)$ from~$M_{i+1}$ to obtain~$M_i$. 
	In the medial graph~$G_{i+1}$, the edge~$e$ corresponds to four edges sharing the vertex corresponding to $e$,
	and there are two unique faces corresponding to $v_1$ and~$v_2$, respectively. 
	Each of those has two of the edges of~$G_{i+1}$ corresponding to $e$ as consecutive edges along the face.
	Let $g_1$ and $g_2$ be those consecutive incident edges on the face of $G_{i+1}$ corresponding to $v_1$.
	Note that their non-shared endpoints lie on the edges $(v_1,v_3)$ and $(v_1,v_4)$, respectively, where~$v_3$ and~$v_4$ are consecutive in the cyclic order around~$v_1$ in $M_i$.
	Further, note that, when removing $e$ from $M_{i+1}$, we have to replace $g_1$ and $g_2$ by an edge~$g$ connecting their non-shared endpoints. 
For every $j \in \{1,2,3,4\}$, let $d(v_j)$ be the disk of~$C(M_p)$ that corresponds to the vertex $v_j$ of $M_i \subset M_p$ (note that with the notation from the invariants, $d(v_1)=d(g_1)=d(g_2)$). 
Next, consider $c(g_1)$ and $c(g_2)$ in the drawing $\Gamma_{i+1}$. 
	By our invariants,  $c(g_1)$ and $c(g_2)$ lie in their lens regions~$\ell(g_1)$ and $\ell(g_2)$, which are consecutive along the boundary of $d(v_1)$.
	The only common point of $\ell(g_1)$ and $\ell(g_2)$ is the touching point of~$d(v_1)$ and $d(v_2)$. 
	The other endpoints of $c(g_1)$ and $c(g_2)$ are the touching points $d(v_1) \cap d(v_3)$ and $d(v_1) \cap d(v_4)$, respectively.
	Further, the boundary of $d(v_1)$ is completely covered by lens regions which are all pairwise non-intersecting and bounded by circles intersecting~$\partial d(v_1)$ in right angles. 
We replace~$c(g_1)$ and~$c(g_2)$ by the circular arc $c(g)$ that has as its endpoints at the touching points $d(v_1) \cap d(v_3)$ and $d(v_1) \cap d(v_4)$ and is tangent to~$c(g_1)$ and~$c(g_2)$, respectively, in its endpoints. 
			We define the lens region~$\ell(g)$ as the unique region that contains~$\ell(g_1)$ and $\ell(g_2)$ and is the intersection of~$d(v_1)$ with the (according side of the) unique disk~$d'(g)$ for which~$\partial d'(g)$ intersects~$\partial d(v_1)$ at a right angle in the endpoints of $c(g)$; see~Fig.~\ref{fig:lens}.

	\begin{figure}[t]
		\centering
		\begin{subfigure}[b]{.47\textwidth}
			\centering
			\includegraphics[scale=1, page=16]{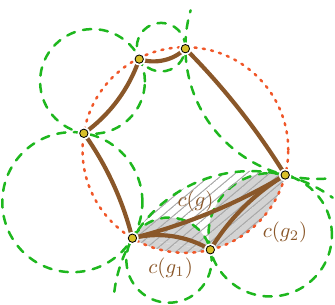}
			\caption{}\label{sfg:lens1}
		\end{subfigure}
		\hfil
		\begin{subfigure}[b]{.47\textwidth}
			\centering
			\includegraphics[scale=1, page=17]{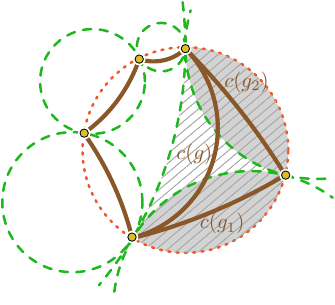}
			\caption{}\label{sfg:lens2}
		\end{subfigure}			
		\caption{Two examples of a lens region $\ell(g)$ resulting from $\ell(g_1)$ and $\ell(g_2)$: (a)~convex and (b)~reflex. The lens regions of $c(g_1)$ and $c(g_2)$ are drawn as shaded areas, while the one of~$c(g)$ is the cross-hatched region.
}
		\label{fig:lens}
	\end{figure}

	Note that $\ell(g)$ does not intersect the interior of any other lens region: for the lens regions outside $d(v_1)$, this is trivial. For the ones inside $d(v_1)$, it follows from continuous transformation of the bounding circle $\partial d'(g)$ to the bounding circle of the other lens. 
	Hence, after repeating the analogous construction for the two other edges in $G_{i+1}$ needed to be replaced when removing $e$ from $M_{i+1}$, 
	namely the ones that are incident to the face corresponding to $v_2$ in $D(G_{i+1})$, 
	we obtain a plane Lombardi drawing $\Gamma_i$ that again fulfills our four invariants, which completes the proof.
\end{proof}

We remark that this result is not tight: there exist 4-regular plane multigraphs whose primal-dual pair $M$ and $M'$ contain parallel edges that still admit plane Lombardi drawings, e.g., knots $8_{12}$, $8_{14}$, $8_{15}$, $8_{16}$; see Fig.~\ref{fig:primes} in Appendix~\ref{app:smallknots}.

We now prove that 4-regular polyhedral graphs are medial graphs
of a simple primal-dual pair.

\begin{lemma}\label{lem:3connsimple}
  Let~$G=(V,E)$ be a 4-regular polyhedral graph
  and let~$M$ and~$M'$ be the primal-dual pair for which~$G$ is the medial graph.
  If there is a multi-edge in~$M$ or in~$M'$, then the corresponding
  vertices of~$G$ either have a multi-edge between them or they form a separation 
  pair of~$G$.
\end{lemma}

\begin{figure}[b]
	\centering
  \hfill
  \subcaptionbox{\label{fig:multiprimaldual-1}}{\includegraphics[scale = 1, page=1]{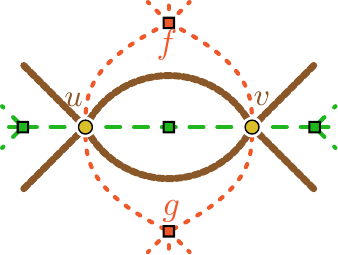}}
  \hfill
  \subcaptionbox{\label{fig:multiprimaldual-2}}{\includegraphics[scale = 1, page=2]{multiprimaldual}}
  \hfill
  \hfill
	\caption{If there is a multi-edge between vertices~$f$ and~$g$ in the primal,
  then there is a multi-edge~$(u,v)$ or a separation pair~$u,v$ in the medial.}
  \label{fig:multiprimaldual}
\end{figure}

\begin{proof}
  \wlogc, assume that there are two edges between 
  vertices~$f$ and~$g$ in~$M$. Let~$u$ and~$v$ be the vertices of~$G$ that
  these two edges pass through; see Fig.~\ref{fig:multiprimaldual}. 
  The vertices~$f$ and~$g$ of~$M$ correspond
  to faces in the embedding of~$G$ that both contain~$u$ and~$v$. Hence, the
  removal of~$u$ and~$v$ from~$G$ disconnects~$G$ into two parts: the
  part inside the area spanned by the two edges between~$f$ and~$g$ and
  the part outside this area. Both~$u$ and~$v$ have two edges in both areas, 
  so either there is a multi-edge between~$u$ and~$v$, or there are vertices in 
	both parts, which makes~$u,v$ a separation pair of~$G$.
\end{proof}

Lemma~\ref{lem:3connsimple} and Theorem~\ref{thm:polyhedral} immediately give 
the following theorem.

\begin{theorem}\label{thm:3connsimple}
  Let~$G=(V,E)$ be a 4-regular polyhedral graph. Then~$G$
  admits a plane Lombardi drawing.
\end{theorem}

\section{Positive and Negative Results for Small Graphs}\label{sec:negative}

We next consider all prime knots with 8 vertices or less. We compute plane Lombardi drawings for those that have it and argue that such drawings do not exists for the others. We start by showing that no knot with a $K_4$ subgraph is plane Lombardi.

\begin{lemma}\label{lem:4knot}
	Every 4-regular plane multigraph $G$ that contains $K_4$ as a subgraph does 
not admit a plane Lombardi drawing. 
\end{lemma}

\begin{proof}
  Let~$a,b,c,d$ be the vertices of the~$K_4$. Every plane embedding of~$K_4$
  has a vertex that lies inside the cycle through the other~3 vertices;
  let~$d$ be this vertex. Since~$d$ has degree~4, it has another edge
  to either one of~$a,b,c$, or to a different vertex. In the former case,
  assume that there is a multi-edge between~$c$ and~$d$. In the latter case,
  by 4-regularity, there has to be another vertex of~$a,b,c$ that is
  connected to a vertex inside the cycle through~$a,b,c$; let~$c$ be this vertex.
  In both cases,~$c$ has two edges that lie inside the cycle through~$a,b,c$.
  
  Assume that~$G$ has a Lombardi drawing.
  Since Möbius transformations do not change the properties of a Lombardi 
  drawing, we may assume that the edge~$(a,b)$ is drawn 
  as a straight-line segment as in Fig.~\ref{sfg:4_1_lom}. Since both~$c$ 
  and~$d$ are neighbors of $a$ and $b$, there are two corresponding placement 
  circles by Property~\ref{prop:placement}. In fact, since any two edges of a 
  Lombardi drawing of a 4-regular graph must enclose an angle of $90^\circ$ and since $a$ 
  and $b$ have \enquote{aligned tangents} due to being neighbors themselves, 
  the two placement circles coincide and a situation as shown in Fig.~\ref{fig:4_1_fails} 
  arises. In particular, this means that in any Lombardi drawing of~$G$ 
  the four vertices must be co-circular. It is easy to see 
  that we cannot draw the missing circular arcs connecting $c$ and $d$: any such arc 
  must either lie completely inside or completely outside of the placement 
  circle. Yet, the stubs for the two edges between~$c$ and~$d$ point inside 
  at $c$ and outside at $d$.
\end{proof}

	\begin{figure}[b]
      \centering
      \begin{subfigure}[b]{.37\linewidth}
        \centering
        \includegraphics[page=2]{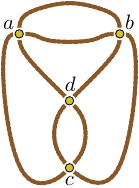}
        \caption{}\label{sfg:4_1}
      \end{subfigure}
      \hfill
      \begin{subfigure}[b]{.57\linewidth}
        \centering
        \includegraphics[page=3]{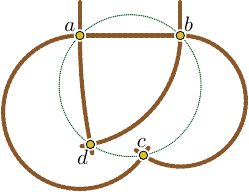}
        \caption{}\label{sfg:4_1_lom}
      \end{subfigure}			
      \caption{Knot $4_1$ has no Lombardi drawing.}
      \label{fig:4_1_fails}
	\end{figure}

	\begin{figure}[t]
      \centering
      \begin{subfigure}[b]{.47\linewidth}
        \centering
        \includegraphics[page=2]{knot5_2_not}
        \caption{}\label{sfg:5_2}
      \end{subfigure}
      \hfill
      \begin{subfigure}[b]{.47\linewidth}
        \centering
        \includegraphics[page=1]{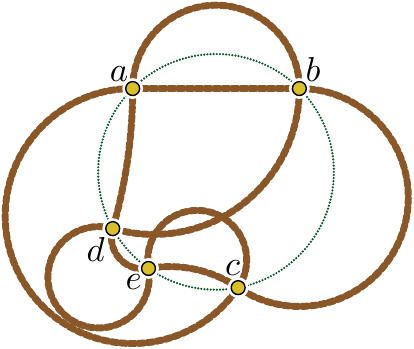}
        \caption{}\label{sfg:5_2_lom}
      \end{subfigure}			
      \caption{Knot $5_2$ and a non-plane Lombardi drawing.}
      \label{fig:5_2}
	\end{figure}

\begin{lemma}\label{lem:4_1}
	Knots $4_1$ and $5_2$ are not plane Lombardi knots.
\end{lemma} 

\begin{proof}
	For knot~$4_1$, the claim immediately follows from Lemma~\ref{lem:4knot}.

	Knot $5_2$ again has the property that all five vertices must be co-circular in any Lombardi drawing. 
	To see this, we first consider the four vertices $a,b,c,d$ in Fig.~\ref{fig:5_2}. 
	Regardless of the placement of $a$ and $b$, we observe that $c$ and $d$ are both adjacent to $a$ and $b$ and need to enclose an angle of $90^\circ$ in the triangular face with~$a$ and~$b$. 
	This situation was already discussed in Lemma~\ref{lem:4knot} and yields a circle $C$ containing $a,b,c,d$; see Fig.~\ref{fig:4_1_fails}.
	The final vertex,~$e$, is adjacent to~$c$ and~$d$ so that we can determine the placement circle for $e$ with respect to~$c$ and~$d$. 
	As we know from Lemma~\ref{lem:4knot}, the two arc stubs of $d$ to be connected with~$e$ form angles of $45^\circ$ with~$C$ and point outwards. 
	Conversely, the two arc stubs of $c$ form angles of $45^\circ$ with~$C$ and point inwards. 
	If we take any point $p$ on~$C$ and draw circular arcs from the stubs of $c$ and $d$ to $p$, the four arcs meet at~$90^\circ$ angles in~$p$. 
	These are precisely the angles required at vertex $e$ and hence~$C$ is in fact the unique placement circle for $e$ by Property~\ref{prop:placement}.
This implies that actually all five vertices of $5_2$ must be co-circular in any Lombardi drawing.

	Unlike knot $4_1$, it is geometrically possible to draw all edges as Lombardi arcs; see Fig.~\ref{sfg:5_2_lom}. 
	However, as we will show, no plane Lombardi drawing of knot~$5_2$ exists. 
	By an appropriate Möbius transformation, we may assume that all five vertices are collinear on a circle of infinite radius. 
	Moreover, to avoid crossings, the order along the line is either $a,b,c,e,d$ or $a,b,d,e,c$ (modulo cyclic shifts and reversals). 
	Since both cases are symmetric, we restrict the discussion to the first one. 
	As a further simplification, we initially assume that $a$ and $b$ are placed on the same position such that the lens between $a$ and $b$ collapses; see Fig.~\ref{fig:5_2_overlap}.

	\begin{figure}[b]
		\centering
			\includegraphics[scale=1,page=3]{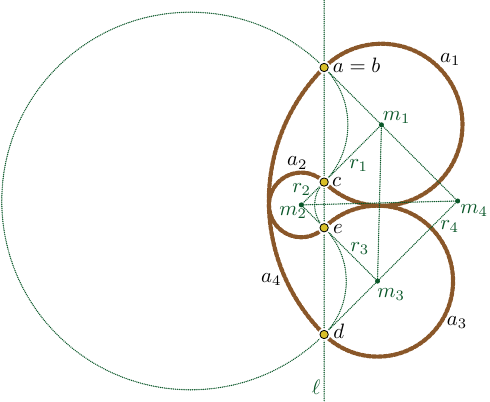}
		\caption{Knot $5_2$ has no plane Lombardi drawing.}
		\label{fig:5_2_overlap}
	\end{figure}

	This drawing consists of two intertwined 4-cycles, which intersect the line~$\ell$ at angles of $45^\circ$. 
	We argue that the 4-cycle depicted in Fig.~\ref{fig:5_2_overlap} cannot be drawn as a simple cycle without self intersections. 
	We consider the four centers $m_1, m_2, m_3, m_4$ of the circular arcs $a_1, a_2, a_3, a_4$ and their radii $r_1, r_2, r_3, r_4$. 
	Due to the fact that adjacent arcs meet on $\ell$ at an angle of $45^\circ$ and have the same tangent, 
	the four centers form the corners of a rectangle~$R$ with side lengths $r_1+r_2$ and $r_3+r_2$. 
	We can further derive that $r_4-r_3 = r_1+r_2$. Let $\delta$ be the length of a diagonal of~$R$. 
	For the arcs $a_1$ and $a_3$ to be disjoint, we require $\delta > r_1 + r_3$. 
	For $a_2$ and $a_4$ to be disjoint, we require $\delta < r_4 - r_2$. 
	But since $r_4 - r_2 = r_1 + r_3$, this is impossible and the 4-cycle must self-intersect. 

	Finally, if we move $b$ by some $\varepsilon > 0$ away from $a$ and towards $c$, 
	this will only decrease the radius $r_4$ and thus introduce proper intersections in the drawing. 
	Thus, knot $5_2$ has no plane Lombardi drawing.
\end{proof}

As the above lemma shows, even very small knots may not have a plane Lombardi drawing. 
However, most knots with a small number of crossings are indeed plane Lombardi.
In Fig.~\ref{fig:primes} in Appendix~\ref{app:smallknots}, we provide plane Lombardi drawings of all knots with up to eight vertices except~$4_1$ and~$5_2$. 
Most of these drawings can actually be obtained using the techniques from Section~\ref{sec:general} and~\ref{sec:circlepacking}.

\begin{theorem}\label{thm:smallknots}
	All prime knots with up to eight vertices other than $4_1$ and $5_2$ are plane Lombardi knots.
\end{theorem}

Note that Theorem~\ref{thm:smallknots} implies that each of these knots has
a combinatorial embedding that supports a plane Lombardi drawing. It is not true, however,
that every embedding admits a plane Lombardi drawing. In fact, the knot $7_5$, as a member of an infinite family of knots and links, has 
an embedding that cannot be drawn plane Lombardi. 
This family is derived from the knot $5_2$ and gives rise to the next theorem.

\begin{theorem}\label{thm:inf_family}
    There exists an infinite family of prime knots and links 
    that have embeddings that do not support plane Lombardi drawings.
\end{theorem}

\begin{proof}
    Consider again the knot $5_2$ (Fig.~\ref{fig:impossible_knots-1}). By Lemma~\ref{lem:4_1}, it has no Lombardi drawing. 
	We claim that if we duplicate the bottom vertex and vertically detach the two copies completely, 
	the resulting graph (using four stubs to ensure the correct angular resolution) 
    still has no plane Lombardi drawing (Fig.~\ref{fig:impossible_knots-2}). 
    As a result, we can construct an infinite family of forms of knots and links without plane Lombardi drawings by vertically twisting the connections between the duplicated vertices.
	The first two smallest members of this 
    family are the link~$L_6a_1$, consisting of two interlinked figure-8's 
    (Fig.~\ref{fig:impossible_knots-3}), and the knot $7_5$ 
    (Fig.~\ref{fig:impossible_knots-4}). \end{proof}

\begin{figure}[t]
	\centering
  \subcaptionbox{$5_2$\label{fig:impossible_knots-1}}{\includegraphics[page=2]{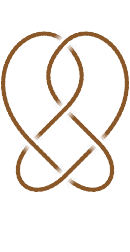}}
  \hfill
  \subcaptionbox{\label{fig:impossible_knots-2}}{\includegraphics[page=3]{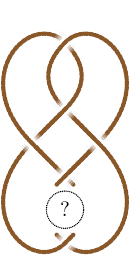}}
  \hfill
  \subcaptionbox{$L_6 a_1$\label{fig:impossible_knots-3}}{\includegraphics[page=4]{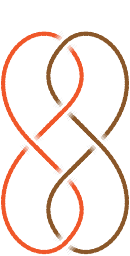}}
  \hfill
  \subcaptionbox{$7_5$\label{fig:impossible_knots-4}}{\includegraphics[page=5]{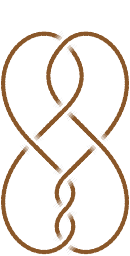}}
	\caption{A family of non-Lombardi knots and links.}
	\label{fig:impossible_knots}
\end{figure}

However, the same family, starting with its six-vertex member $L_6a_1$, does have plane Lombardi drawings with a different embedding.

\begin{corollary}\label{cor:family_lombardi}
	The prime knots and links in the family of Theorem~\ref{thm:inf_family} with six or more vertices all have an embedding with a plane Lombardi drawing.
\end{corollary}

\begin{figure}[tb]
	\centering
    \subcaptionbox{$L_6 a_1$}{\includegraphics[page=3]{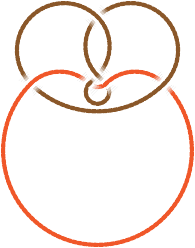}}
    \hfill
    \subcaptionbox{$7_5$}{\includegraphics[page=4]{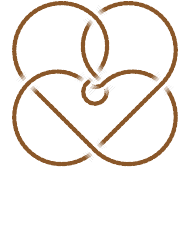}}
    \hfill
    \subcaptionbox{$L_8 a_{11}$}{\includegraphics[page=5]{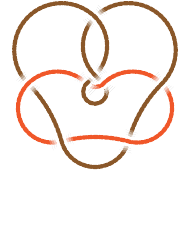}}
    \hfill
    \subcaptionbox{$9_6$}{\includegraphics[page=6]{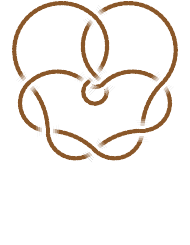}}
		
	\caption{Plane Lombardi drawings for Corollary~\ref{cor:family_lombardi} obtained by lens multiplication.}
	\label{fig:familiy_lombardi}
\end{figure}

\begin{proof}
	For the link $L_6a_1$ and the knot $7_5$ we provide plane Lombardi drawings in Fig.~\ref{fig:familiy_lombardi}.
	Observe that the incremental twists defining the family in Theorem~\ref{thm:inf_family} are now done at the bottom part of the knot/link diagrams in Fig.~\ref{fig:familiy_lombardi}. Since each twist now corresponds to a lens multiplication, we obtain from Lemma~\ref{lem:lensmult} that all other knots and links in the family also have plane Lombardi drawings. Figure~\ref{fig:familiy_lombardi} shows the respective 8-vertex link $L_8a_{11}$ and 9-vertex knot $9_6$.
\end{proof}

Interestingly, the two knots $4_1$ and $5_2$ without a plane Lombardi drawing belong to the known family of \emph{twist knots}, which are knots formed by taking a closed loop, twisting it any number of times and then hooking up the two ends together. All other twist knots do have plane Lombardi drawings though.

\begin{corollary}
	All twist knots except $4_1$ and $5_2$ are plane Lombardi knots.
\end{corollary}

\begin{proof}
	We know from Theorem~\ref{thm:smallknots} that $4_1$ and $5_2$ are not plane Lombardi knots.
	The smallest twist knot is $3_1$ and the other twist knots with at most eight vertices are $6_1$, $7_2$, and $8_1$, which all have a plane Lombardi drawing as shown in Fig.~\ref{fig:primes}. The progressive twisting  pattern defining the twist knots and seen in the drawings of $6_1$, $7_2$, and $8_1$ can easily be extended for all twist knots by incrementally applying the lens multiplication of Lemma~\ref{lem:lensmult}.
\end{proof}

\section{Plane 2-Lombardi Drawings of Knots and Links}
\label{sec:2lombardi}

Since not every knot admits a plane Lombardi drawing, we now consider plane 
2-Lombardi drawings; see Fig.~\ref{fig:knot41-smooth} for an example.
Bekos et al~\cite{Smoothorthogonalintro} recently introduced 
\emph{smooth orthogonal drawings of complexity $k$}. 
These are drawings where every edge consists of a sequence of at most $k$ 
circular arcs and axis-aligned segments that meet smoothly with horizontal or
vertical tangents, and where at every 
vertex, each edge emanates either horizontally or vertically and no two edges 
emanate in the same direction. For the special case of 4-regular graphs,
every smooth orthogonal drawing of complexity~$k$ is also a plane $k$-Lombardi drawing.
Alam et al.~\cite{Smoothorthogonal4planar} showed that every plane graph with maximum degree 4 can be
redrawn as a plane smooth-orthogonal drawing of complexity 2. 
Their algorithm takes as input an orthogonal drawing produced by the algorithm 
of Liu et al.~\cite{lms-la2be-DAM98} and transforms it
into a smooth orthogonal drawing of complexity~2. 
We show how to modify the algorithm by Liu et al., to compute an orthogonal drawing for
a 4-regular plane multigraph and then use the algorithm by Alam et al.\ to transform
it into a smooth orthogonal drawing of complexity~2.

\begin{figure}[t]
	\centering
	\subcaptionbox{\label{fig:knot41-smooth}}{\includegraphics[page=2,scale=1, angle=90]{knot4_1}}
  \hfil
  \subcaptionbox{\label{fig:knot41-2lomb}}{\includegraphics[page=1,scale=1]{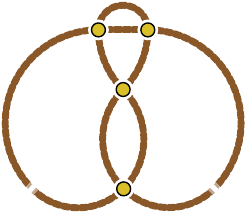}}
  \hfil
	\subcaptionbox{\label{fig:knot41-epslomb}}{\includegraphics[scale=1.0]{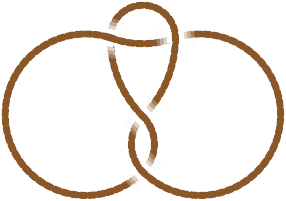}}
	\caption{Drawings of knot $4_1$ which by Lemma~\ref{lem:4_1} does not admit a plane Lombardi drawing. (a)~A smooth orthogonal drawing of complexity~2, (b)~a different plane 2-Lombardi drawing, and (c)~a plane $\eps$-angle Lombardi drawing.}
  \label{fig:2lombardi}
\end{figure}

\begin{theorem}\label{thm:smoothorthogonal}
  Every biconnected 4-regular plane multigraph $G$ admits a plane 2-Lombardi drawing with the same embedding.
\end{theorem}

\begin{proof}
  The algorithm of Alam et al.~\cite{Smoothorthogonal4planar} takes as input
  an orthogonal drawing produced by the algorithm of Liu et al.~\cite{lms-la2be-DAM98} and transforms it
  into a smooth orthogonal drawing of complexity~2. The drawings by Liu et al.\ 
  have the property that every edge consists of at most~3 segments (except
  at most one edge that has~4 segments), and it contains no \emph{S-shapes}, that is, it contains
  no edge that consists of~3 segments where the bends are in opposite direction.
  To show this theorem, we only have to show that we can apply the algorithm
  of Liu et al.\ to 4-regular plane multigraphs to produce a drawing with the
  same property.
  
  Liu et al.\ first choose two vertices~$s$ and~$t$ and compute an 
  \emph{$st$-order} of the input graph. An $st$-order is an ordering 
  $(s=1,2,\dots,n=t)$ of the vertices such that every~$j$ ($2<j<n-1$) 
  has neighbors $i$ and $k$ with $i<j<k$. We can obtain an $st$-order for a
  multigraph by removing any duplicate edges. Liu et al.\ then direct all
  edges according to the $st$-order from a vertex with lower $st$-number to a
  vertex with higher $st$-number.
  
  According to the rotation system implied by the embedding of the input graph,
  Liu et al.\ then assign a port to every edge around a vertex such that
  every vertex (except~$t$) has an outgoing edge at the top port, every vertex
  (except~$s$) has an incoming edge at the top port, every vertex has an
  outgoing edge at the right port if and only if it has at least~2 outgoing
  edges, and every vertex has an incoming edge at the left port if and only
  if it has at least~2 incoming edges. They further make sure 
  the edge that uses the bottom port at~$s$
  is incident to the vertex~$r$ with $st$-number~2, and that the 
  edge~$(s,t)$, if it exists, uses the left port at~$s$ and the top port at~$t$;
  this edge is the only one drawn with~4 segments, but can still be transformed
  into a smooth orthogonal edge of complexity~2 by Alam el al.\ .
  They place the vertices~$s$ and~$r$ on the $y$-coordinate~2 and every
  other vertex on the $y$-coordinates
  equal to their $st$-number. The shape of the edges is then 
  implied by the assigned ports at their incident vertices. By placing vertices
  that share an edge with a bottom port and a top port above each other,
  there can be no S-shapes with two vertical segments, but there can still be
  S-shapes with two horizontal segments if an edge uses a left port and a right 
  port. To eliminate these S-shapes, the consider sequences of S-shapes, that
  is, paths in the graphs that are drawn only with S-shapes, and move the
  vertices vertically such that they all lie on the same $y$-coordinate.
  Up to the elimination of S-shapes, every step of the algorithm can 
  immediately applied to multigraphs. We choose~$s$ and~$t$ as vertices
  on the outer face of the given embedding such that the edge~$(s,t)$
  exists. We claim that then no multi-edge can be drawn as an S-shape.
  
  Let~$u$ and~$v$ be two vertices in~$G$ with at least two edges~$e_1$ and~$e_2$
  between them. \wlogc, let~$u$ have a lower $st$-number than~$v$.
  Then both~$e_1$ and~$e_2$ are directed from~$u$ to~$v$. 
  If~$u=s$ and~$v=r$, then both vertices are placed on the same $y$-coordinate,
  so there can be no S-shape between them. If $u=s$ and~$v=t$, then there
  is an edge that uses the left port at~$u$ and the top port at~$v$;
  since all multi-edges have to be consecutive around~$u$ and~$v$, there
  can be no edge between them that uses a left port and a right port.
  Otherwise, assume that~$e_2$ is the successor of~$e_1$ in counter-clockwise 
  order around~$u$ (and hence the predecessor of~$e_1$ in counter-clockwise 
  order around~$v$). If~$e_1$ uses the right port at~$u$ and the left port 
  at~$v$, then~$e_2$ has to use the top port at~$v$, which cannot occur by
  the port assignment. If~$e_1$ uses the left port at~$u$ and the right port 
  at~$v$, then~$e_2$ has to use the bottom port at~$u$, which also cannot
  occur by the port assignment. Thus, neither~$e_1$ nor~$e_2$ is drawn as 
  an S-shape and every sequence of S-shapes consists only of simple edges.
  Hence, we can use the algorithm of Liu et al.\ to produce an orthogonal
  drawing with the desired property for every 4-regular plane multigraph and
  then use the algorithm of Alam et al.\ to transform it into a smooth
  complexity drawing of complexity~2 which is also a plane 2-Lombardi drawing.
\end{proof}

\section {Plane Near-Lombardi Drawings}
\label{sec:almost}
  
  Since not all knots admit a plane Lombardi drawing, in this section we relax the 
  perfect angular resolution constraint. We say that a knot (or a link) is
  \emph{near-Lombardi} if it admits a drawing for every $\eps>0$ such that
  \begin{enumerate}[nolistsep]
    \item All edges are circular arcs,
    \item Opposite edges at a vertex are tangent;
    \item The angle between crossing pairs at each vertex is at least $90^\circ - \eps$. 
  \end{enumerate}
  We call such a drawing a \emph{$\eps$-angle Lombardi drawing}. Note that a
  Lombardi drawing is essentially a $0$-angle Lombardi drawing. For example, the
  knot~$4_1$ does not admit a plane Lombardi drawing, but it admits a plane $\eps$-angle
  Lombardi drawing, as depicted in Fig.~\ref{fig:knot41-epslomb}.
  
Let~$\Gamma$ be an $\eps$-angle Lombardi drawing of a 4-regular graph.
  If each angle described by the tangents of adjacent circular arcs at a vertex 
  in~$\Gamma$ is exactly~$90^\circ+\eps$ or~$90^\circ-\eps$, then we call~$\Gamma$
  an \emph{$\eps$-regular Lombardi drawing}. Note that any Lombardi drawing is a
  $0$-regular Lombardi drawing.
  
  We first extend some of our results for plane Lombardi drawings to
  plane $\eps$-angle Lombardi drawings. The following Lemma is a
  stronger version of Theorem~\ref{thm:polyhedral}.

  \begin{lemma}\label{lem:polyhedraleps}
    Let~$G=(V,E)$ be a biconnected 4-regular plane multigraph and let~$M$ and~$M'$
    be the primal-dual multigraph pair for which~$G$ is the medial graph. If 
    one of~$M$ and~$M'$ is simple, then~$G$ admits a plane $\eps$-regular 
    Lombardi drawing preserving its embedding for every $0^\circ\le\eps<90^\circ$. 
  \end{lemma}

  \begin{proof}
    We use the same algorithm as for the proof of Theorem~\ref{thm:polyhedral} with
    a slight modification. We first seek to direct the edges such that
    every vertex has two incoming opposite edges and two outgoing opposite edges.
    Let~$M$ and~$M'$ be the primal-dual pair corresponding to the medial graph~$G$.
    Every face in~$G$ corresponds to a vertex either in~$M$ or in~$M'$; we say
    that the face \emph{belongs to}~$M$ or~$M'$. We orient the edges around
    each face that belongs to~$M$ in counter-clockwise order. Every edge 
    in~$G$ lies between a face that belongs to~$M$ and a face that belongs 
    to~$M'$, so this gives a unique orientation for every edge. Further,
    the faces around any vertex belong to~$M$, to~$M'$, to~$M$, and to~$M'$
    in counter-clockwise order. Hence, the edges around any vertex are 
    outgoing, incoming, outgoing, and incoming in counter-clockwise order,
    which gives us the wanted edge orientation.
    
    We use the same primal-dual circle packing approach to obtain a drawing
    of~$G'$, but instead of using the bisection of the intersection of a
    primal and a dual circle, we use a circular arc with a different angle;
    see Fig.~\ref{fig:polyhedraleps-1}. Let~$e=(u,v)$ be an edge of~$G'$
    directed from~$u$ to~$v$, and
    let~$l(e)$ be the lens region of~$e$ between the primal-dual 
    circles~$d(e)$ and~$d'(e)$. \wlogc, assume that the $90^\circ$ angle
    inside~$l(e)$ is between~$d'(e)$ and~$d(e)$ in counter-clockwise order around~$u$.
    In the proof of Theorem~\ref{thm:polyhedral}, we would draw~$e$ as a 
    bisection of~$l(e)$.   
    We draw~$e$ that the angle between~$d'(e)$ and~$(u,v)$ at~$u$
    is~$45^\circ+\eps/2$ and the angle between~$e$ and~$d(e)$ at~$u$ 
    is~$45^\circ-\eps/2$. 
  
  \begin{figure}[t]
    \begin{subfigure}[b]{.47\textwidth}
      \centering
      \includegraphics[page=1]{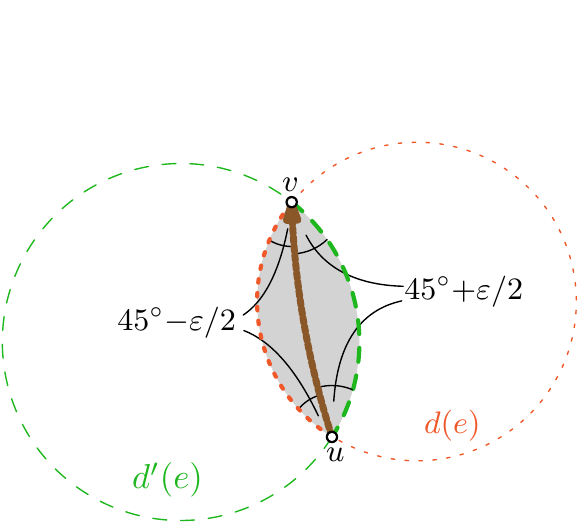}
      \caption{Drawing a directed edge~$e=(u,v)$ between the circles~$d(e)$ and~$d'(e)$}
      \label{fig:polyhedraleps-1}
    \end{subfigure}
    \hfil
    \begin{subfigure}[b]{.47\textwidth}
      \centering
      \includegraphics[page=2]{3connsimple-epsreg}
      \caption{Eliminating a vertex~$v$ and adding the edge~$(u,w)$ inside the circle~$d(e_1)$}
      \label{fig:polyhedraleps-2}
    \end{subfigure}
    \caption{Illustrations for the proof of Lemma~\ref{lem:polyhedraleps}.}
    \label{fig:polyhedraleps}
  \end{figure}
    
    Informally, this means that all outgoing edges 
    at a vertex are ``rotated'' by $\eps/2$ in counter-clockwise direction, and
    all incoming edges at a vertex are ``rotated'' by $\eps/2$ in clockwise 
    direction compared to a plane circular-arc drawing of~$G'$. Since opposite 
    edges of a vertex~$u$ have the same direction with respect to~$u$, 
    they are rotated by the same angle, so they are still tangent. Further, 
    since adjacent edges at~$u$ have a different direction with respect to~$u$,
    the angle between them is now either~$90^\circ+\eps$ or~$90^\circ-\eps$.
    
    We then use the same procedure as in Theorem~\ref{thm:polyhedral} to eliminate
    vertices from~$G'$ and obtain a plane $\eps$-regular Lombardi drawing of~$G$.
    In every step of this procedure, we eliminate a vertex~$v$ from~$G'$ and add
    an edge between two pairs of its adjacent vertices (without introducing
    self-loops); see Fig.~\ref{fig:polyhedraleps-2}. 
    Let~$u$ and~$w$ be two neighbors of~$v$ in~$G'$ such that we
    want to obtain the edge~$e=(u,w)$ in~$G$. \wlogc, assume
    that the edge~$e_1=(u,v)$ is directed from~$u$ to~$v$ in~$G'$ and that the
    edge~$e_2=(v,w)$ is directed from~$v$ to~$w$ in~$G'$. Following 
    the proof of Theorem~\ref{thm:polyhedral}, $e_1$ lies in the
    lens region $l(e_1)$ between disks $d(e_1)$ and~$d'(e_1)$, and~$e_2$ lies in the
    lens region $l(e_1)$ between disks $d(e_2)=d(e_1)$ and~$d'(e_2)$. Hence,~$u$ 
    and~$w$ lie on a common circle~$d(e_1)$ of the primal-dual circle packing. 
    Assume that the $90^\circ$ angle inside~$l(e_1)$ is between~$d'(e_1)$ 
    and~$d(e_1)$ in counter-clockwise order around~$u$; the other case is 
    symmetric. By the direction of
    the edges~$e_1$ and~$e_2$, the angle between~$e_1$ and~$d(e_1)$ 
    is~$45^\circ-\eps/2$ in counter-clockwise around~$u$ 
    and the angle between~$d(e_1)$ and~$e_2$ is 
    also~$45^\circ-\eps/2$ in counter-clockwise direction around~$w$. Hence, we can
    draw the edge~$e$ as a circular arc inside~$d(e_1)$ with angle~$45^\circ-\eps/2$
    to~$d(e_1)$ at both~$u$ and~$w$. We keep the ports at both vertices and
    by directing the edge from~$u$ to~$w$ we also keep a direction of the edges
    that satisfies the above property. Thus, we obtain a plane  $\eps$-regular
    Lombardi drawing of~$G$.
  \end{proof}

  The following Lemmas are stronger versions of Lemma~\ref{lem:lensmult} and Theorem~\ref{thm:summation}, respectively. 
  Since the proofs of the latter results do not rely on~$90^\circ$ angles, 
  they can also applied to the stronger versions.
For the sake of completeness, a formal proof of Lemma~\ref{lem:lensmulteps} is still given.

  \begin{lemma}\label{lem:lensmulteps}
    Let $G=(V,E)$ be a 4-regular plane multigraph with a plane $\eps$-angle Lombardi 
    drawing~$\Gamma$. Then, any lens multiplication $G'$ of $G$ also admits a 
    plane $\eps$-angle Lombardi drawing.
  \end{lemma}

  \begin{proof}
    Let $f$ be a lens in $\Gamma$ spanned by two vertices $u$ and $v$. We 
    denote the two edges bounding the lens as $e_1$ and $e_2$. 
    Let~$\alpha\in[90^\circ-\eps,90^\circ+\eps]$ be
    the angle between $e_1$ and $e_2$ in both end-vertices. We define the bisecting 
    circular arc $b$ of~$f$ as the unique circular arc connecting $u$ and $v$ 
    with an angle of $\alpha/2$ to both~$e_1$ and $e_2$. See 
    Fig.~\ref{fig:lens-subdivision} for an example.

    Let $p$ be the midpoint of $b$. If we draw circular arcs $a_1$ and $a_2$
    from both $u$ to~$p$ and circular arcs~$a_3$ and~$a_4$ from $v$ to $p$ 
    that have the same tangents as $e_1$ and $e_2$ in $u$ and $v$, then these 
    four arcs meet at~$p$ such that the angle between~$a_1$ and~$a_2$ as well
    as the angle between~$a_3$ and~$a_4$ is~$\alpha$, whereas the angle 
    between~$a_1$ and~$a_4$ and the angle between~$a_2$ and~$a_3$ 
    is~$180^\circ-\alpha\in[90^\circ-\eps,90^\circ+\eps]$.
    Further, each such arc lies inside lens $f$ and hence does not cross 
    any other arc of $\Gamma$. The resulting drawing is thus a plane $\eps$-angle Lombardi 
    drawing of a 4-regular multigraph that is derived from $G$ by subdividing the 
    lens~$f$ with a new degree-4 vertex.
	
    By repeating this construction inside the new lenses we can create plane 
    $\eps$-angle Lombardi drawings that replace lenses by chains of smaller lenses.  .
  \end{proof}

  \begin{lemma}\label{lem:summationeps}
    Let $A$ and $B$ be two 4-regular plane multigraphs with plane $\eps$-angle Lombardi drawings. 
    Let $a$ be an edge of $A$ and $b$ an edge of $B$. Then the composition $A + B$
    obtained by connecting $A$ and $B$ along edges $a$ and $b$ admits a plane 
    $\eps$-angle Lombardi drawing.
  \end{lemma}
  
  Let $G=(V,E)$ be a 4-regular plane multigraph 
  and let~$x\in V$ with edges~$(x,a)$, $(x,b)$, $(x,c)$, and~$(x,d)$
  in counter-clockwise order. A \emph{lens extension} of~$G$ is a 
  4-regular plane multigraph that is obtained by removing~$x$ and its incident
  edges from~$G$, and adding two vertices~$u$ and~$v$ to~$G$ with two edges 
  between~$u$ and~$v$ and the edges $(u,a),(u,b),(v,c),(v,d)$.
  Informally, that means that a vertex is substituted by a lens.
  
  \begin{lemma}\label{lem:lensext}
    Let $G=(V,E)$ be a 4-regular plane multigraph with a plane $\eps$-angle Lombardi drawing~$\Gamma$. 
	Then, any lens extension of~$G$ admits a plane $(\eps+\eps')$-angle Lombardi drawing for every~$\eps'>0$.
  \end{lemma}

  \begin{proof}
    Let~$x\in V$ be the vertex that we want to perform the lens extension on
    such that we get the edges~$(u,a),(u,b),(v,c),(v,d)$
    in the obtained graph~$G'$. Let~$\alpha$ be the angle between the tangents
    of~$(x,a)$ and~$(x,b)$ at~$x$ in~$\Gamma$. Since~$\Gamma$ is a plane
    $\eps$-angle Lombardi drawing, we have that~$\alpha\le 90^\circ+\eps$.
    Further, the angle between the tangents of~$(x,c)$ and~$(x,d)$ at~$x$ 
    in~$\Gamma$ is also~$\alpha$, while the angles between the tangents of~$(x,b)$ 
    and~$(x,c)$ at~$x$ and between the tangents of~$(x,d)$ and~$(x,a)$ at~$x$ 
    are both~$180^\circ-\alpha$. We apply the Möbius-transformation on~$\Gamma$
    that maps the edges~$(x,a)$ and~$(x,d)$ to straight-line segments
    and~$a$ lies on the same~$y$-coordinate and to the right of~$x$;
    hence,~$d$ lies strictly below~$x$.
  
  \begin{figure}[t]
    \centering
    \begin{subfigure}[b]{.45\textwidth}
      \centering
      \includegraphics[width=\textwidth,page=1]{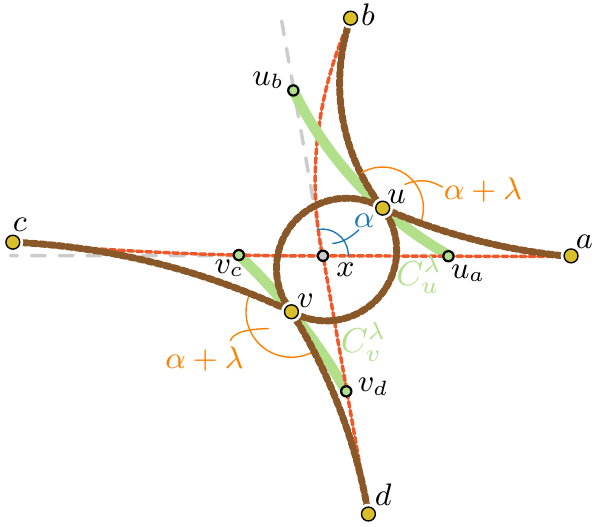}
      \caption{}
      \label{fig:lens-extension-1-app}
    \end{subfigure}
    \hfill
    \begin{subfigure}[b]{.45\textwidth}
      \centering
      \includegraphics[width=\textwidth,page=2]{lens-extension-all}
      \caption{}
      \label{fig:lens-extension-2}
    \end{subfigure}
    \caption{(a)~The circular arc~$C_u^\lambda$ between~$u_a$ and~$u_b$ 
      on the placement circles of~$u$ and the circular arc~$C_v^\lambda$ 
      between~$v_c$ and~$v_d$ on the placement circles of~$v$.
      (b)~Placing~$u$ on~$u_a$ and~$v$ on~$v_c$ gives~$\beta_a=0^\circ$. }
    \label{fig:lens-extension-12}
  \end{figure}
    
We aim to place~$v$ such that the angle between the arcs~$(v,c)$ and~$(v,d)$ 
    is $\alpha+\lambda$ for some~$0<\lambda\le\eps'$ which we will show how
    to choose later. We have fixed ports at~$c$ and~$d$
    and a fixed angle $\alpha+\lambda$ at~$v$. According to 
    Property~\ref{prop:placement}, all possible 
    positions of~$v$ lie on a circle through~$c$ and~$d$. Note that the circle
    through~$c,d,x$ describes all possible positions of neighbors of~$c$ and~$d$
    with angle~$\alpha$. Since the desired angle gets larger, the position circle
    for~$v$ contains a point~$v_d$ on the straight-line edge~$(x,d)$ and a 
    point~$v_c$ on the half-line starting from~$x$ with the angle of the port
    used by the arc~$(x,c)$; see Fig.~\ref{fig:lens-extension-1-app}.
    We denote by~$C_v^\lambda$ the circular arc between~$v_c$ and~$v_d$ on the
    placement circle of~$v$ that gives the angle~$\alpha+\lambda$ at~$v$. We do 
    the same construction for~$u$ to obtain the circular 
    arc~$C_u^\lambda$ between~$u_a$ and~$u_b$.
    
    Since the drawing of~$G$ is plane, there is some non-empty region in which
    we can move~$x$ such that the arcs $(x,a),(x,b),(x,c),(x,d)$ are drawn with
    the same ports at~$a,b,c,d$ and do not cross any other edge of the drawing.
    We choose~$\lambda$ as the largest value with~$0<\lambda\le\eps'$
    such that the two circular arcs~$C_u$ and~$C_v$ lie completely inside this
    region. 
    
    We now have to find a pair of points on~$C_v^\lambda$ and~$C_u^\lambda$ 
    such that we can connect them via a lens. 
    The ports of the two arcs we seek to draw between~$u$ and~$v$ lie opposite
    of the ports used by the arcs~$(u,a),(u,b),(v,c)$, and~$(v,d)$
    We label the ports at~$u$ and~$v$ 
    as~$p_u^a$ opposite of $(u,a)$ at~$u$, as~$p_u^b$ opposite of $(u,b)$ at~$u$,
    as~$p_v^c$ opposite of $(v,c)$ at~$v$, and as~$p_v^d$ opposite of $(v,d)$ at~$v$.
    We have to find a pair of points on~$C_u^\lambda$ and~$C_v^\lambda$ such that
    these ports are ``compatible'':
    Take a point~$q_u$ on~$C_u^\lambda$ and a point~$q_v$ on~$C_v^\lambda$ and 
    connect them by a segment~$S_{uv}$.
    Then the angle~$\beta_b$ between~$S_{uv}$ and $p_u^b$ has to be the same as the angle~$\beta_c$
    between~$S_{uv}$ and $p_v^c$, and the angle~$\beta_a$ between~$S_{uv}$ and $p_u^a$ has 
    to be the same as the angle~$\beta_d$ between~$S_{uv}$ and $p_v^d$. 
    By construction, we have that~$\beta_a+\beta_b=90^\circ+\alpha+\lambda$ 
    and~$\beta_c+\beta_d=90^\circ+\alpha+\lambda$, so
    it suffices to find a pair of points such that~$\beta_a=\beta_d$.
    
    Assume that~$v$ is placed on~$v_c$ and~$u$ is placed on~$u_a$; see Fig.~\ref{fig:lens-extension-2}.
    The edge~$(x,a)$ is drawn as a straight-line segment, and the edge~$(x,c)$
    uses the port opposite of the one of~$(x,a)$. Hence, the segment~$S_{uv}$
    is a segment through~$x$. Furthermore, it uses exactly the port~$p_u^a$
    at~$u$, so we have~$\beta_a=0^\circ$. On the other hand,~$\beta_d$ is 
    strictly positive: The segment~$(x,d)$ enters~$x$ with an angle
    of~$\gamma=180^\circ-\alpha>0^\circ$ to the segment~$(x,v)$. 
    Since~$v$ lies to the left of~$x$, the angle
    described between the tangent of the circular arc~$(v,d)$ at~$v$ and the
    segment~$(v,x)$ is strictly larger than~$\gamma$. Since~$\beta_d$ is 
    described by the same tangent and segment, we have that~$\beta_d=\gamma>0^\circ$.
  
  \begin{figure}[t]
    \centering
    \begin{subfigure}[b]{.45\textwidth}
      \centering
      \includegraphics[width=\textwidth,page=3]{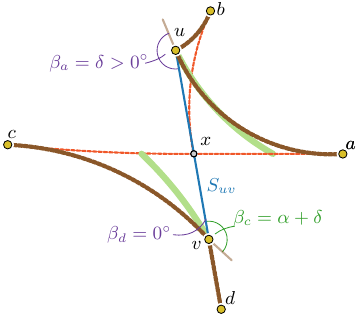}
      \caption{}
      \label{fig:lens-extension-3}
    \end{subfigure}
    \hfill
    \begin{subfigure}[b]{.45\textwidth}
      \centering
      \includegraphics[width=\textwidth,page=4]{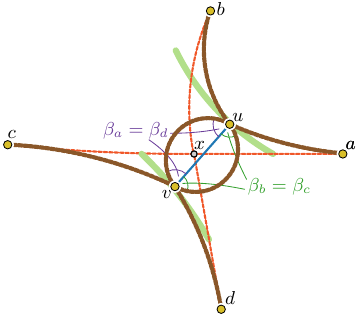}
      \caption{}
      \label{fig:lens-extension-4}
    \end{subfigure}
    \caption{(a)~Placing~$u$ on~$u_b$ and~$v$ on~$v_d$ gives~$\beta_d=0^\circ$.
      (b)~Placing~$u$ and~$v$ such that~$\beta_a=\beta_d$ gives a lens
      between~$u$ and~$v$ with the desired angles.}
    \label{fig:lens-extension-34}
  \end{figure}
    
    Now assume that~$v$ is placed on~$v_d$ and~$u$ is placed on~$u_b$; see Fig.~\ref{fig:lens-extension-3}.
    The edge~$(x,d)$ is drawn as a straight-line segment, and the edge~$(x,b)$
    uses the port opposite of the one of~$(x,d)$. Hence, the segment~$S_{uv}$
    is a segment through~$x$. Furthermore, it uses exactly the port~$p_u^d$
    at~$u$, so we have~$\beta_d=0^\circ$. On the other hand,~$\beta_a$ is 
    strictly positive: The segment~$(x,a)$ enters~$x$ with an angle
    of~$\delta=90^\circ+\alpha>0^\circ$ to the segment~$(x,v)$. 
    Since~$u$ lies above~$x$, the angle
    described between the tangent of the circular arc~$(u,a)$ at~$u$ and the
    segment~$(u,x)$ is strictly larger than~$\delta$. Since~$\beta_d$ is 
    described by the same tangent and segment, we have that~$\beta_d=\delta>0^\circ$.
    
    Hence, we have found a pair of points for~$u$ and~$v$ such that~$\beta_a=0^\circ$
    and~$\beta_d=\gamma>0^\circ$ and we have found a pair of points for~$u$ and~$v$ such 
    that~$\beta_a=\delta>0^\circ$ and~$\beta_d=0\circ$. Since we can move~$u$ and~$v$ freely
    along the curves~$C_u$ and~$C_v$ between these pairs of points,~$\beta_a$
    can become any angle between~$0^\circ$ and~$\delta$ and~$\beta_d$ can
    become any angle between~$0^\circ$ and~$\gamma$. Thus, there has to exist
    some pair of points for~$u$ and~$v$ such that~$\beta_a=\beta_d$; see Fig.~\ref{fig:lens-extension-4}.
    We choose this pair of points and connect~$u$ and~$v$ by two circular arcs
    such that one of them uses the ports~$p_u^a$ and~$p_v^d$ and the other
    one uses the ports~$p_u^b$ and~$p_v^c$. Note that the arcs~$(u,a)$ and~$(u,b)$
    are now drawn the same way as if we moved~$x$ onto the determined position
    of~$u$ and the arcs~$(v,c)$ and~$(v,d)$
    are now drawn the same way as if we moved~$x$ onto the determined position
    of~$v$. Hence, by the choice of~$\lambda$, they do not introduce any
    crossing and thus the drawing is plane.
  \end{proof}

  \begin{lemma}\label{lem:maxdeg3}
    Every 4-regular plane multigraph with
    at most~3 vertices admits a plane $\eps$-regular Lombardi drawing for
    every $0\le\eps<90^\circ$.
  \end{lemma}

  \begin{proof}
    There are only two 4-regular multigraphs with at most~3 vertices and each of them
    has a plane Lombardi drawing as depicted in Fig.~\ref{fig:biconn3lomb}. 
    For some~$0^\circ<\eps<90^\circ$, we can obtain a plane $\eps$-regular Lombardi
    drawing by simply making the circular arcs larger or smaller, as depicted in
    Fig.~\ref{fig:biconn3eps}.
  \end{proof}

   \begin{figure}[t]
    \centering
\begin{subfigure}[b]{.4\textwidth}
        \centering
		\includegraphics[page=1]{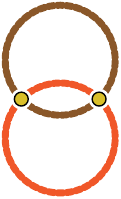} \hspace{1cm}
\includegraphics[page=2]{biconn3}
        \caption{}
        \label{fig:biconn3lomb}
      \end{subfigure}
      \hfill
      \begin{subfigure}[b]{.55\textwidth}
        \centering
        \includegraphics[page=3]{biconn3} \hspace{1cm}
\includegraphics[page=4]{biconn3}
        \caption{}
        \label{fig:biconn3eps}
      \end{subfigure}
      \caption{The only biconnected 4-regular multigraphs with at most~3 vertices.
        (a) plane Lombardi and (b) plane $\eps$-angle Lombardi drawings.}
      \label{fig:biconn3}
\end{figure}
 
  We are now ready to present the main result of this section.
  The proof boils down to a large case distinction using the tools developed in the 
  previous discussion. We split the original graph into biconnected components
   and then use Lemma~\ref{lem:maxdeg3} and~\ref{lem:polyhedraleps} as base cases.
  With the help of lens extensions, lens multiplications, and knot sums we can combine the \enquote{near-Lombardi} drawings 
  of the biconnected graphs to generate an \enquote{near-Lombardi} drawing of the original graph.
  As a consequence, every knot is near-Lombardi.

  \begin{theorem}\label{thm:biconneps}
    Let $G=(V,E)$ be a biconnected 4-regular plane multigraph and let~$\eps>0$.
    Then~$G$ admits a plane $\eps$-angle Lombardi drawing.
  \end{theorem}

\begin{proof} 
    If~$G$ has at most~3 vertices, then we obtain a plane Lombardi drawing
    of~$G$ by Lemma~\ref{lem:maxdeg3}.    
    So assume that~$G$ is a biconnected 4-regular plane graph with~$n\ge 4$.
    We seek to draw~$G$ by recursively by splitting it into smaller graphs.
    We prove our algorithm by induction on the number of vertices; to this end, 
    suppose that every biconnected 4-regular plane graph with at most~$n-1$
    vertices admits a plane $\eps'$-angle Lombardi drawing for every~$\eps'>0$;
    this holds initially for $n=4$. We proceed as follows.
    
    \ccase{c:3con-app} $G$ is polyhedral. In this case, we can draw it 
    plane Lombardi using Theorem~\ref{thm:polyhedral}. 
    
    \begin{figure}[b]
      \centering
      \subcaptionbox{Case~\ref{c:multilens-app}\label{fig:c:multilens-app}}{\includegraphics{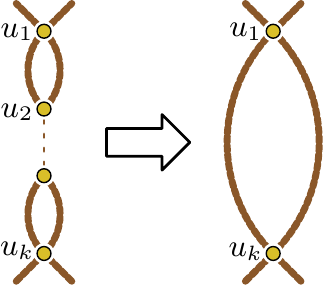}}
      \hfil
      \subcaptionbox{Case~\ref{sc:lens-3-app}\label{fig:sc:lens-3-app}}{\includegraphics{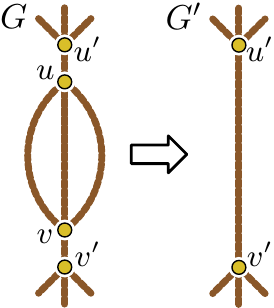}}
      \caption{Illustrations for the proof of Theorem~\ref{thm:biconneps}.}
      \label{fig:biconneps-1}
    \end{figure}
    
    \ccase{c:multilens-app} $G$ contains a \emph{multilens}, that is, a sequence of 
    lenses between the vertices~$u_1,\ldots,u_k$ with~$k\ge 2$.
    We contract the lenses to a single lens,
    that is, we remove the vertices~$u_2,\ldots,u_{k-1}$ 
    and their incident edges from~$G$ and add two edges between~$u_1$ and~$u_k$
    to form a new graph~$G'$; see Fig.~\ref{fig:c:multilens-app}.
    This operation is essentially a reverse lens multiplication and introduces 
    no self-loops. It also preserves biconnectivity since~$u_1$ and~$u_k$ 
    form a separation pair in~$G$, so any cutvertex in~$G'$ would also be
    a cutvertex in~$G$. Hence,~$G'$ is a biconnected 4-regular plane graph
    with~$n-k+1\le n-1$ vertices and by induction admits a plane $\eps$-angle
    Lombardi drawing. Furthermore,~$G$ is a lens multiplication of~$G'$
    on the lens~$(u_1,u_k)$, so we can use Lemma~\ref{lem:lensmulteps}
    to obtain a plane $\eps$-angle Lombardi drawing of~$G$.
    
    \ccase{c:lens-app} $G$ contains a lens between two vertices~$u$ and~$v$, but 
    it contains no multilens. We consider three subcases based on the number
    of edges between~$u$ and~$v$ in~$G$.
    
    \subcase{sc:lens-4-app} There are four edges between~$u$ and~$v$ in~$G$.
    Since~$G$ is 4-regular, it consists exactly of these two vertices and four 
    edges and can be drawn by Lemma~\ref{lem:maxdeg3}.
    
    \subcase{sc:lens-3-app} There are three edges between~$u$ and~$v$ in~$G$; see Fig.~\ref{fig:sc:lens-3-app}.
    Then there exists also some edge~$(u,u')$ and some edge~$(v,v')$ in~$G$. 
    Since~$G$ is biconnected, we have~$u'\neq v'$; otherwise,
    it would be a cutvertex. We remove~$u$ and~$v$ from~$G$ and add
    an edge between~$u'$ and~$v'$ to form a new graph~$G'$. 
    This operation preserves biconnectivity 
    as~$u'$ and~$v'$ form a separation pair in~$G$
    and it introduces no self-loops because $v\neq v'$.
    Hence, the graph~$G'$ is a biconnected 4-regular plane graph with~$n-2$
    vertices and by induction admits a plane $\eps$-angle Lombardi drawing.
    Let~$G''$ be the graph that consists of~$u$ and~$v$ and four multi-edges 
    between them. This graph has a plane $\eps$-regular Lombardi drawing 
    by Lemma~\ref{lem:maxdeg3}.
    Furthermore,~$G$ can be obtained by adding~$G'$ and~$G''$
    along the edge~$(u',v')$ of~$G'$ and one of the edges of~$G''$.
    Using Lemma~\ref{lem:summationeps}, we can obtain a plane $\eps$-angle Lombardi drawing of~$G$.
    
    \subcase{sc:lens-2-app} There are two edges between~$u$ and~$v$ in~$G$.
    We consider two subcases.
    
    \subsubcase{ssc:lens-2-nosep-app} Removal of~$u$ and~$v$ from~$G$ preserves
    connectivity; see Fig.~\ref{fig:ssc:lens-2-nosep-app}. We contract~$u$ and~$v$ to a new vertex:
    we remove them from~$G$ and add a new vertex~$x$ that is connected to the 
    neighbors of~$u$ and~$v$ different from~$u$ and~$v$ to form~$G'$. This operation 
    preserves biconnectivity: Since~$G$ is biconnected, the only cutvertex in~$G'$
    can be~$x$; but since the removal of~$u$ and~$v$ from~$G$ preserves connectivity,
    so does the merged vertex~$x$. Since there are exactly two edges between~$u$
    and~$v$, the new vertex~$x$ has degree~4. Hence,~$G'$ is a biconnected
    4-regular plane graph with~$n-1$ vertices and by induction admits a plane
    $\eps/2$-angle Lombardi drawing. Furthermore,~$G$ can be obtained from~$G'$
    by a lens extension on~$x$. We obtain a plane $\eps$-angle Lombardi drawing
    of~$G$ using Lemma~\ref{lem:lensext}.
    
    \begin{figure}[t]
      \centering
      \subcaptionbox{Case~\ref{ssc:lens-2-nosep-app}\label{fig:ssc:lens-2-nosep-app}}{\includegraphics[page=1]{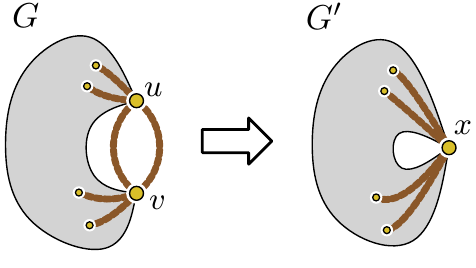}}
      \hfill
      \subcaptionbox{Case~\ref{ssc:lens-2-sep-app}\label{fig:ssc:lens-2-sep-app}}{\includegraphics[page=2]{case-lens-2}}
      \caption{Illustrations for Case~\ref{sc:lens-2-app} in the proof of Theorem~\ref{thm:biconneps}.}
      \label{fig:lens-2}
    \end{figure}
    
    \subsubcase{ssc:lens-2-sep-app} The removal of~$u$ and~$v$ from~$G$ disconnects the
    graph, that is,~$u$ and~$v$ form a separation pair in~$G$; see Fig.~\ref{fig:ssc:lens-2-sep-app}. Since there
    are exactly two edges between~$u$ and~$v$, their removal disconnects~$G$
    into two connected components~$A$ and~$B$ with at least two vertices each
    (otherwise, there would be a self-loop). Furthermore, $G$ contains an edge
    from~$u$ to a vertex~$u_A$ in~$A$ and another edge from~$u$ to a vertex~$u_B$ in~$B$. 
    If this would not be the case than~$v$ would be a cutvertex in~$G$. 
    Analogously, there is an edge from~$v$ to a vertex~$v_A$
    in~$A$ and an edge from~$v$ to a vertex~$v_B$ in~$B$.
    We have~$u_A\neq v_A$ (and $u_B\neq v_B$); otherwise, this vertex would be a cutvertex in~$G$.
    
    Let~$A'$ be the graph~$G-B$ with an additional edge between~$u$ and~$v$.
    Since~$G$ is biconnected, there are two disjoint paths in~$G$ between any two 
    vertices from~$A$. Only one of these paths can \enquote{leave}~$A$ through 
    the separation pair $u,v$. Hence, we can redirect the part outside~$A$ to 
    the new edge $(u,v)$ in~$A'$, which shows that every two vertices in~$A'$ 
    are connected with at least two disjoint paths. This shows that~$A'$ is
    biconnected.
    
    Let~$B'$ be the graph~$B$ with an additional edge between~$u_B$ and~$v_B$.
    We can show that~$B'$ is biconnected by the same arguments we have applied
    for~$A'$: In~$G$ there have to be two disjoint paths between every vertex 
    pair from~$B$. Only one of these paths can leave $B$ over the separation
    pair~$u,v$ and this part can be replaced by the new edge that we added to $B'$.
    Hence between every two vertices in~$B'$ we have two disjoint paths, which 
    proves that~$B'$ is a biconnected 4-regular plane graph with at most $n-4$ vertices.
    By induction~$B'$ admits a plane $\eps$-angle Lombardi drawing.
    Furthermore,~$G$ can be obtained by adding~$A'$ and~$B'$
    along one of the edges between~$u$ and~$v$ of~$A'$ and the edge~$(u_B,v_B)$
    of~$B'$. Using Lemma~\ref{lem:summationeps}, we can obtain a plane $\eps$-angle 
    Lombardi drawing of~$G$.
    
    \ccase{c:nolens-app} $G$ is simple, but not 3-connected, so there exists at least
    one separation pair that splits~$G$ into at least two connected components.
    Let~$A_{u,v}$ be a smallest connected component induced by the separation
    pair~$u,v$. We say that~$u,v$ is a \emph{minimal separation pair} if
    $A_{u,v}$ does not contain any separation pair and there is no 
    separation pair between a vertex of~$A_{u,v}$ and either~$u$ or~$v$.
    
    \begin{figure}[t]
      \centering
      \subcaptionbox{Case~\ref{c:nolens-app}\label{fig:c:nolens-app}}{\includegraphics[page=1]{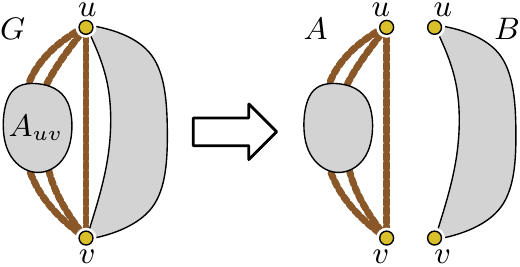}}
      \hfil
      \subcaptionbox{Case~\ref{sc:nolens-33-app}\label{fig:sc:nolens-33-app}}{\includegraphics[page=2]{case-nolens}}
      \caption{Illustrations for Case~\ref{c:nolens-app} in the proof of Theorem~\ref{thm:biconneps}.}
      \label{fig:biconneps-2}
    \end{figure}
    
    We create two biconnected 4-regular plane plane graphs as
    follows; see Fig.~\ref{fig:c:nolens-app}. Let~$A$ be the subgraph of~$G$ induced by the vertices in~$A_{u,v}$, 
    $u$, and~$v$, let~$B$ be the subgraph of~$G$ that contains all vertices
    not in~$A_{u,v}$ and all edges not in~$A$; in particular, there is no edge~$(u,v)$ in~$B$.
    By this construction, all edges of~$G$ are either part of~$A$ or part of~$B$
    and both~$A$ and~$B$ are connected, and every vertex is part of either~$A$
    or~$B$, except the two vertices~$u$ and~$v$ which are part of both.
    However,~$A$ and~$B$ are not 4-regular, so we create two 4-regular 
    graphs~$A'$ and~$B'$ for the recursion as follows. 
    Let~$\deg_A(u),\deg_A(v),\deg_B(u),\deg_B(v)$ be the degree of~$u$ and~$v$
    in~$A$ and~$B$, respectively, with~$\deg_A(u)+\deg_B(u)=\deg_A(v)+\deg_B(v)=4$.
    
    \subcase{sc:nolens-11-app} $\deg_A(u)=1$ or $\deg_A(v)=1$. \wlogc,
    let~$\deg_A(u)=1$. Let~$x$ be the neighbor of~$u$
    in~$A$. We have that~$x\neq v$ since otherwise~$A$ consists only of a
    single edge (if~$\deg_A(v)=1$) or~$v$ is a cutvertex in~$G$ (if~$\deg_A(v)=3$).
    Then~$x,v$ is a separation
    pair of~$G$ whose removal gives a connected component~$A_{x,v}$
    with less vertices than~$A_{u,v}$, as it contains the same vertices but not~$x$,
    contradicting the minimality of the separation pair~$u,v$.
    
    \subcase{sc:nolens-33-app} $\deg_A(u)=\deg_A(v)=3$; see Fig.~\ref{fig:sc:nolens-33-app}. 
    We add an edge between~$u$ and~$v$ to~$A$
    to obtain the graph~$A'$. The resulting graph is biconnected: consider
    any pair of vertices~$a,b\in A'$. There were at least two vertex-disjoint
    paths in~$G$ between~$a$ and~$b$. Since~$u,v$ is a separation pair in~$G$,
    at most one of these two paths traverses vertices in~$G-A$, and
    any path through these vertices must contain~$u$ and~$v$. Hence, there is a
    path that traverses the same edges in~$A'$ and uses the newly introduced
    edge between~$u$ and~$v$ instead.
    
    We remove~$u$ and~$v$ from~$B$ and add an edge
    between their neighbors to form~$B'$.  Let~$x$ be the neighbor of~$u$
    in~$B$ and let~$y$ be the neighbor of~$v$ in~$B$. We have that~$x\neq y$ 
    since otherwise~$x$ would be a cutvertex in~$G$.
    Hence, we introduce no self-loops. With a similar argument, $B'$ is
    also biconnected, as any path between two vertices through vertices in~$A$
    has to traverse~$u$ and~$v$ and---since they both have degree~1 in~$B$---
    their neighbors, so the path can use the newly introduced edge instead.
    
    We recursively obtain a plane $\eps$-angle Lombardi drawing of~$A'$ and~$B'$.
    Since both~$A'$ and~$B'$ have fewer vertices than~$G$, they admit
    one by induction. To obtain a drawing of~$G$ from~$A'$ and~$B'$, we have to 
    remove the edge~$(u,v)$ 
    from~$A'$ and the edge~$(x,y)$ from~$B'$ and we have to add the edges~$(u,x)$
    and~$(v,y)$. This procedure is equivalent to adding~$A'$ and~$B'$
    along these respective edges,
    so we can solve it using the algorithm described in Lemma~\ref{lem:summationeps}.
    
    \subcase{sc:nolens-22-app} $\deg_A(u)=\deg_A(v)=2$. We consider two more subcases.
    
    \begin{figure}[b]
      \centering
      \includegraphics[page=3]{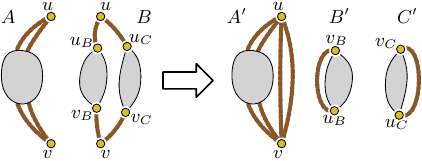}
      \caption{Illustration for Case~\ref{ssc:nolens-22-3comp-app} in the proof of Theorem~\ref{thm:biconneps}.}
      \label{fig:ssc:nolens-22-3comp-app}
    \end{figure}
    
    \subsubcase{ssc:nolens-22-3comp-app} The separation pair~$u,v$ splits~$G$
    into three connected components~$A_{u,v}$, $B_{u,v}$, and~$C_{u,v}$; see Fig.~\ref{fig:ssc:nolens-22-3comp-app}.
    We add two edges between~$u$ and~$v$ to~$A$ to obtain~$A'$.
    Let~$u_B$ be the neighbor of~$u$ in~$B_{u,v}$ and let~$v_B$ be the
    neighbor of~$v$ in~$B_{u,v}$. We have that~$u_B\neq v_B$, as otherwise it would
    be a cutvertex of~$G$. We obtain the 4-regular multigraph~$B'$ by adding an edge between~$u_B$
    and~$v_B$ to~$B_{u,v}$. By the same argument as in Case~\ref{ssc:lens-2-sep-app},
    $B'$ is biconnected. Analogously, we obtain the biconnected 4-regular multigraph~$C'$ by adding an edge
    between the neighbor~$u_C$ of~$u$ in~$C_{u,v}$ and the neighbor~$v_C$ of~$v$
    in~$C_{u,v}$ to~$C_{u,v}$. We recursively create a plane $\eps$-angle Lombardi
    drawing of~$A'$, $B'$, and~$C'$. Then, we create a plane $\eps$-angle Lombardi
    drawing with the use of Lemma~\ref{lem:summationeps} by adding~$A'$ and~$B$
    along one edge between~$u$ and~$v$ of~$A'$ and the edge~$(u_B,v_B)$ of~$B$,
    and adding the resulting graph and~$C$ along the other edge between~$u$ and~$v$ 
    and the edge~$(u_C,v_C)$ of~$C$.
    
    \subsubcase{ssc:nolens-22-2comp-app} The separation pair~$u,v$ splits~$G$
    into two connected components~$A_{u,v}$ and $B_{u,v}$; see Fig.~\ref{fig:ssc:nolens-22-2comp-app}. In this case,
    the graph~$B$ consists of~$B_{u,v}$,~$u$, and~$v$ and the edges incident to~$u$
    or~$v$ and a vertex of~$B_{u,v}$.

    We add two edges between~$u$ and~$v$ to both graphs~$A$ and~$B$ to 
    obtain~$A'$ and~$B'$. 
    Let~$M$ and~$M'$ be the primal-dual pair for which~$A'$ is the medial graph.
    We claim that~$M$ or~$M'$ is simple. 
    Let~$u_1$ and~$u_2$ be the neighbors of~$u$ in~$A$ and 
    let~$v_1$ and~$v_2$ be the neighbors of~$v$ in~$A$.
    Since~$G$ is simple, 
    there is no multi-edge between~$u$ and~$v$ in~$A$.
    Furthermore, there is no single edge~$(u,v)$ in~$A$, since otherwise~$u$ 
    and~$v$ would each only have one neighbor in~$A_{u,v}$ and these neighbors
    would be a separation pair of~$G$ that induces a smaller connected component.
    Thus, each of $u_1,u_2,v_1,v_2$ is different from~$u$ and~$v$ and we introduce no
    self-loops. By construction, the graph~$A_{u,v}$ contains no separation
    pair and thus has either at most~3 vertices or is 3-connected. We claim 
    that~$A'$ is 3-connected. If~$A_{u,v}$ has only~1 vertex, then~$u_1=u_2$,
    so there is a multi-edge in~$A_{u,v}$ which contradicts simplicity of~$G$.
    If~$A_{u,v}$ has only~2 vertices, then there has to be a multi-edge between
    them, which again contradicts simplicity of~$G$.
    If~$A_{u,v}$ has only~3 vertices, then there have to be~4 edges in~$A_{u,v}$,
    which also contradicts simplicity of~$G$.
    If~$A_{u,v}$ has at least~4 vertices, then~$u$ and~$v$
    have at least~3 different neighbors in~$A_{u,v}$, as otherwise there would
    be a cutvertex or a separation pair that gives a smaller connected component
    than the separation pair~$u,v$. Thus, if~$u$ and~$v$ are connected to at 
    least~$3$ vertices of~$A_{u,v}$ 
    and~$u$ and~$v$ are connected by an edge, which preserves
    3-connectivity. Hence,~$A'$ is 3-connected. Since~$G$ is simple, 
    there is no multi-edge between~$u$ and~$v$ in~$A$.
    Furthermore, there is no single edge~$(u,v)$ in~$A$, since otherwise~$u$ 
    and~$v$ would each only have one neighbor in~$A_{u,v}$ and these neighbors
    would be a separation pair of~$G$ that induces a smaller connected component.
    Hence,~$A'$ has no separation pair and exactly one multi-edge between~$u$
    and~$v$. By 
    Lemma~\ref{lem:3connsimple}, that means that~$M$ and~$M'$ have exactly
    one pair of parallel edges in total, so one of them has to be simple.
    
    \begin{figure}[t]
      \centering
      \includegraphics[page=4,width=\textwidth]{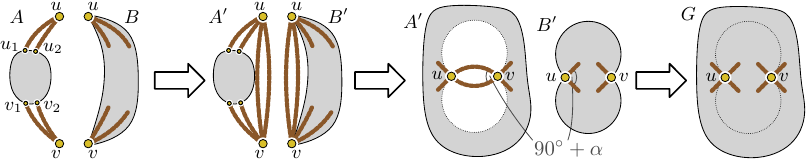}
      \caption{Illustration for Case~\ref{ssc:nolens-22-2comp-app} in the proof of Theorem~\ref{thm:biconneps}.}
      \label{fig:ssc:nolens-22-2comp-app}
    \end{figure}
    
    We recursively obtain a plane $\eps$-angle Lombardi drawing of~$B'$.
    Let~$90^\circ+\alpha$ be the angle described by the tangents of the two edges
    between~$u$ and~$v$ at~$u$. Note that~$\alpha$ might be negative, but $|\alpha|\le\eps$.
    Since the primal or dual of~$A'$ is simple, we obtain a plane $|\alpha|$-regular
    Lombardi drawing of~$A'$ by Lemma~\ref{lem:polyhedraleps}.
    Thus, the angle described by the tangents of the two edges between~$u$ and~$v$
    at~$u$ is either~$90^\circ+\alpha$ or~$90^\circ-\alpha$. We can make sure
    that the angle is~$90^\circ+\alpha$ by inverting the direction of all edges
    in the proof of Lemma~\ref{lem:polyhedraleps} in case it is not.
    
    We perform a Möbius-transformation on the drawing of~$A'$ such that the
    edges between~$u$ and~$v$ are drawn with an angle of~$45^\circ+\alpha/2$
    between either edge and the segment between~$u$ and~$v$. We pick the
    Möbius-transformation such that~$u$ and~$v$ are very close to each other;
    in particular, we want them to be close enough such that the two circles
    that the edges between~$u$ and~$v$ lie on contain no other vertex of~$A'$
    and no edges of~$A'$ that is incident to neither~$u$ nor~$v$. Note that
    the radius of these circles are the same and approach~$0$ as the distance
    between~$u$ and~$v$ approaches~$0$; hence, such a Möbius-transformation exists.
    
    We apply another Möbius-transformation on~$B'$ such that the distance 
    between~$u$ and~$v$ is the same as in the drawing of~$A'$ and such that
    the two edges between~$u$ and~$v$ are drawn with an angle of~$135^\circ-\alpha/2$
    between either edge and the segment between~$u$ and~$v$.
    We now place the drawing of~$B'$ on the drawing of~$A'$ such that both copies
    of~$u$ lie on the same coordinate and both copies of~$v$ lie on the same 
    coordinate and then we remove all edges between~$u$ and~$v$. By construction,
    the whole drawing of~$B'$ lies inside the region described by the two edges
    between~$u$ and~$v$ in the drawing of~$B'$. Further, since these edges lie on the same 
    circles as the two edges between~$u$ and~$v$ in~$A'$, this region contains 
    no vertices or edges in the drawing of~$A'$ (except~$u$ and~$v$ and their
    incident edges themselves). Since the drawings of~$A'$ and~$B'$ are plane
    and we cannot introduce a crossing between an edge of~$A'$ and an edge of~$B'$
    after removing the multi-edges between~$u$ and~$v$, the resulting drawing
    of~$G$ is also plane. Since~$u$ and~$v$ use the same ports in the drawing
    of~$A'$ and the drawing of~$B'$, the resulting drawing is a plane $|\alpha|$-angle
    Lombardi drawing of~$G$. Because of~$|\alpha|\le\eps$, this drawing is also
    a plane $\eps$-angle Lombardi drawing of~$G$.  
\end{proof}

\section{Conclusion and Open Problems}

We have studied plane Lombardi drawings of knots and links, which can
be modeled as 4-regular multigraphs. We have shown
that not all knots admit a plane Lombardi drawing. On the other hand,
we have given an algorithm to draw 4-regular polyhedral multigraphs plane Lombardi.
Further, we have shown that every biconnected 4-regular plane multigraph
admits a plane 2-Lombardi drawing, where every edge is composed of two
circular arcs, and a plane near-Lombardi drawing, where
the angle between two edges at a vertex is at least $90^\circ-\eps$ for any $\eps>0$,
while the angle between opposite edges remains $180^\circ$.

Although we made progress on the original question,
there are several questions that remain open. As main questions concerning Lombardi drawings we have the following.

\begin{question}
Can we give a complete characterization of 4-regular plane multigraphs that admit a plane Lombardi drawing? 
\end{question}

\begin{question}
What is the complexity of deciding whether a given 4-regular plane multigraph admits a plane Lombardi drawing?
\end{question}

\begin{question}
Given a 4-regular plane multigraph, what is the minimum number of edges consisting of two circular arcs in any plane 2-Lombardi drawing?
\end{question}

\paragraph{Acknowledgements.} 
Research for this work was initiated at Dagstuhl Seminar 17072 \emph{Applications of Topology to the Analysis of 1-Dimensional Objects} which took place in February 2017.
We thank Benjamin Burton for bringing the problem to our attention
and Dylan Thurston for helpful discussion.

\bibliographystyle{abbrvurl}
\bibliography{abbrv,bibliography}

\clearpage
\appendix

\section{\texorpdfstring{Drawing Knots~$5_1$, $6_2$, $7_7$, and $8_{18}$ via Circle Packing}{Drawing Knots 5\_1, 6\_2, 7\_7, and 8\_18 via Circle Packing}}\label{app:circle_packings}

\begin{figure}[h]
  \centering
\includegraphics[page=27,width=.4\textwidth]{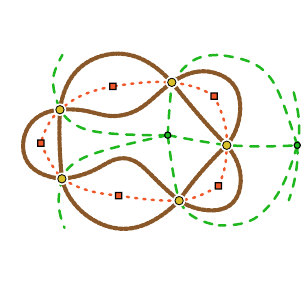}
  \hfill
  \includegraphics[page=12,width=.4\textwidth]{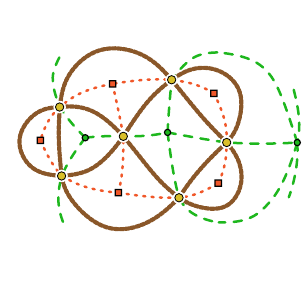}

  \medskip
  \includegraphics[page=13,width=.4\textwidth]{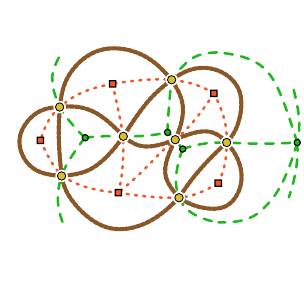}
  \hfill
  \includegraphics[page=14,width=.4\textwidth]{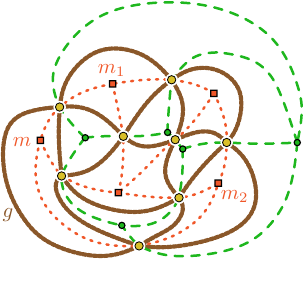}
\caption{Extension of the primal graph (dotted) of knot~$5_1$ to the square pyramid and its dual (dashed). The medial graph in the top right is the knot~$6_2$, the medial graph
in the bottom left is the knot~$7_7$, and the medial graph in the bottom right is the knot~$8_{18}$.}
  \label{fig:extension}
\end{figure}

\begin{figure}[!ht]
  \centering
\includegraphics[page=2,width=.39\textwidth]{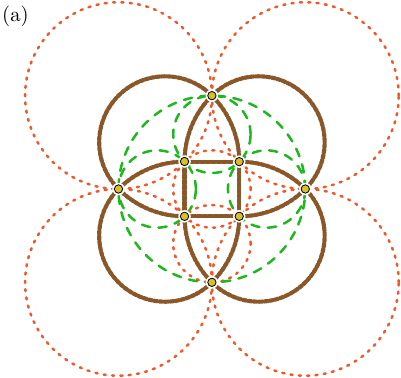}
  \hfil
  \includegraphics[page=8,width=.39\textwidth]{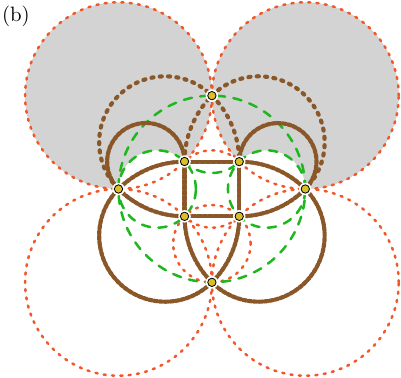}\\[4ex]
  
\includegraphics[page=10,width=.39\textwidth]{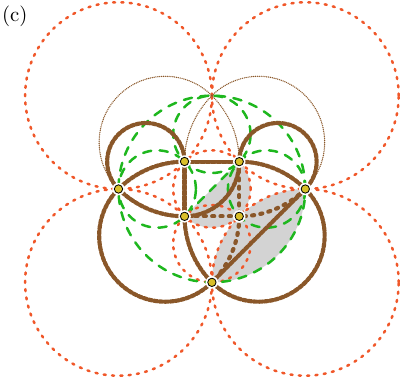}
  \hfil
  \includegraphics[page=11,width=.39\textwidth]{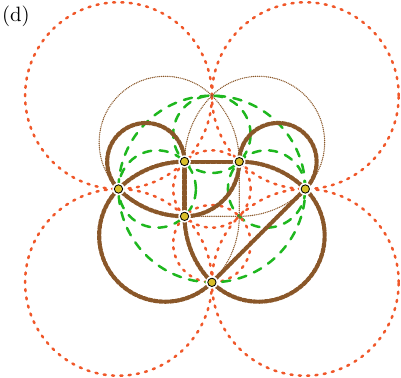}\\[4ex]
  
\includegraphics[page=13,width=.39\textwidth]{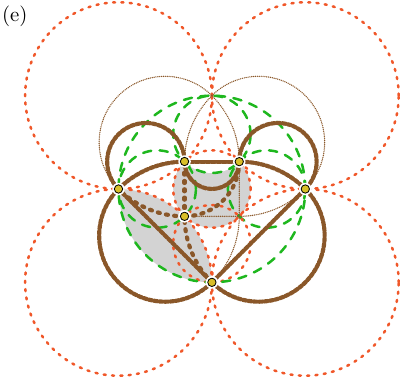}
  \hfil
  \includegraphics[page=15,width=.39\textwidth]{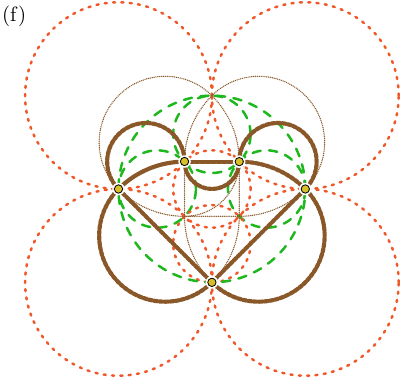}
  
  \bigskip
  \caption{(a)~A circle packing for the square pyramid (dotted) and
    its dual (dashed), and a plane Lombardi drawing for the medial graph $8_{18}$ (solid);
    (b)~eliminating an edge of the primal and the plane Lombardi drawing of~$7_7$;
    (c)~eliminating an edge of the primal; (d)~the plane Lombardi drawing of~$6_2$;
    (e)~eliminating an edge of the primal; (f)~the plane Lombardi drawing of~$5_1$.}
  \label{fig:circle_packings}
\end{figure}

\clearpage
\section{Drawings of all Lombardi Prime Knots up to 8 Vertices}\label{app:smallknots}

	\begin{figure}[htbp]
		\centering
	\begin {tabular} {cccccccc}
	    \includegraphics [scale=.40,page=2] {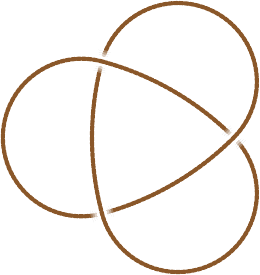} $3_1$
	  & \includegraphics [scale=.40,page=2] {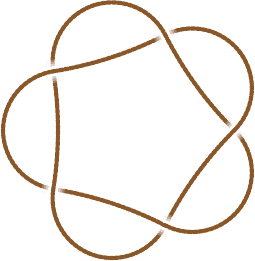} $5_1$
	  & \includegraphics [scale=.40,page=2] {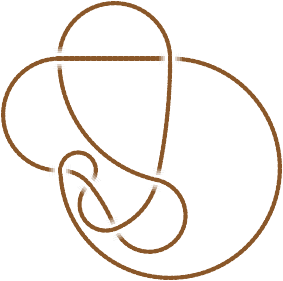} $6_1$
	  & \includegraphics [scale=.40,page=2] {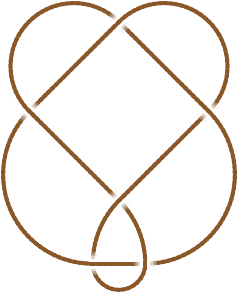} $6_2$
	  & \includegraphics [scale=.40,page=2] {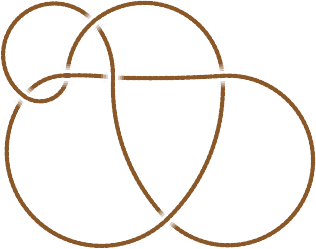} $6_3$ \\[1em]
	  \includegraphics [scale=.40,page=2] {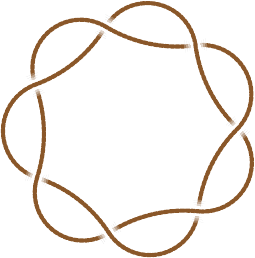} $7_1$  
	  & \includegraphics [scale=.40,page=2] {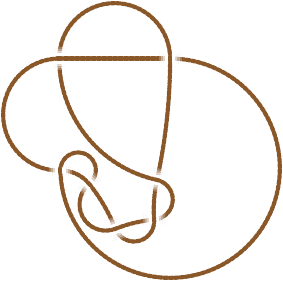} $7_2$  
	  & \includegraphics [scale=.40,page=2] {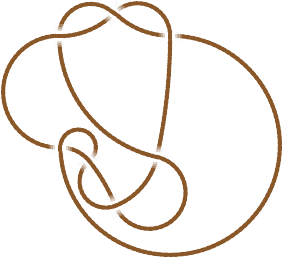} $7_3$  
	  & \includegraphics [scale=.40,page=2] {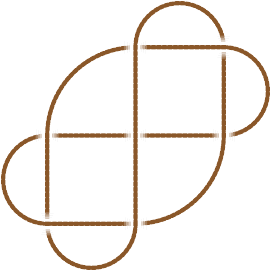} $7_4$  
	  & \includegraphics [scale=.40,page=2] {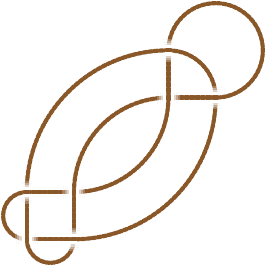} $7_5$  \\[1em]
	  \includegraphics [scale=.40,page=2] {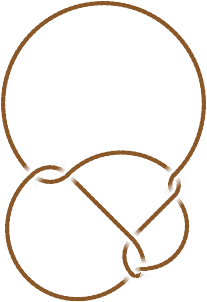} $7_6$  
	  & \includegraphics [scale=.40,page=2] {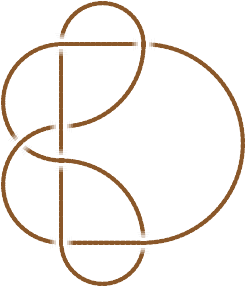} $7_7$  
	  & \includegraphics [scale=.40,page=2] {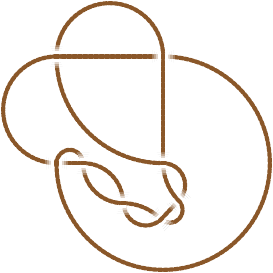} $8_{1}$  
	  & \includegraphics [scale=.40,page=2] {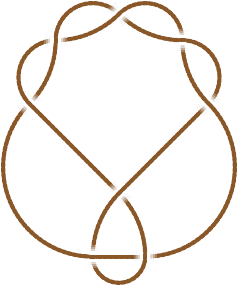} $8_{2}$  
	  & \includegraphics [scale=.40,page=2] {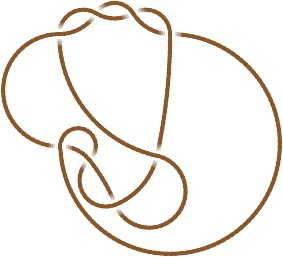} $8_{3}$   \\[1em]
	   \includegraphics [scale=.40,page=2] {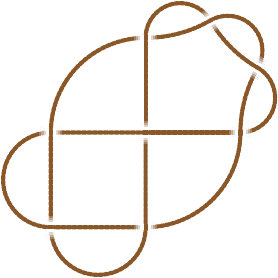} $8_{4}$  
	  & \includegraphics [scale=.40,page=2] {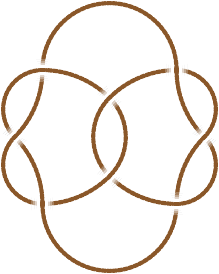} $8_{5}$  
	  & \includegraphics [scale=.40,page=2] {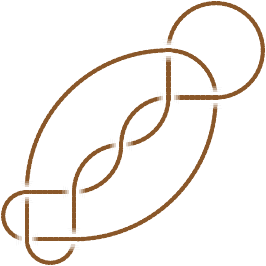} $8_{6}$  
	  & \includegraphics [scale=.40,page=2] {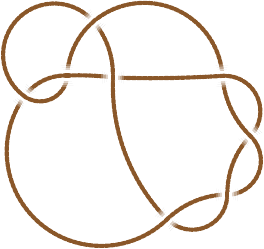} $8_{7}$  
	  & \includegraphics [scale=.40,page=2] {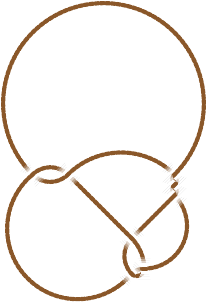} $8_{8}$   \\[1em]
	   \includegraphics [scale=.40,page=2] {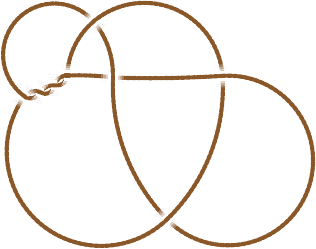} $8_{9}$  
	  & \includegraphics [scale=.40,page=2] {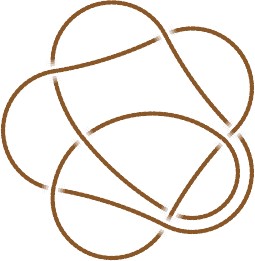} $8_{10}$  
	  & \includegraphics [scale=.40,page=2] {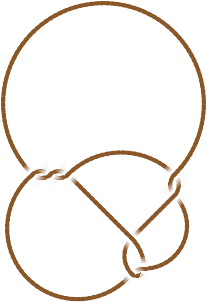} $8_{11}$  
	  & \includegraphics [scale=.40,page=2] {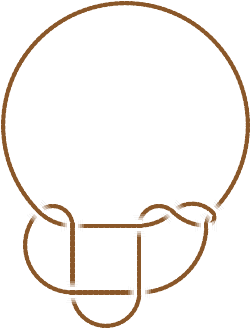} $8_{12}$  
	  & \includegraphics [scale=.40,page=2] {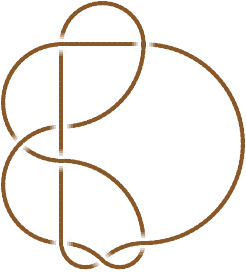} $8_{13}$   \\[1em]
	   \includegraphics [scale=.40,page=2] {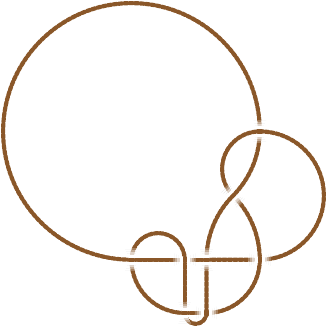} $8_{14}$  
	  & \includegraphics [scale=.40,page=2] {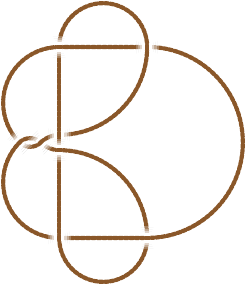} $8_{15}$  
	  & \includegraphics [scale=.40,page=2] {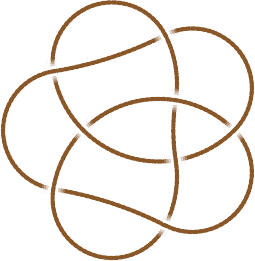} $8_{16}$  
	  & \includegraphics [scale=.40,page=2] {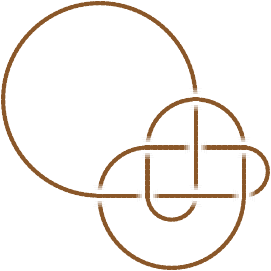} $8_{17}$  
	  & \includegraphics [scale=.40,page=2] {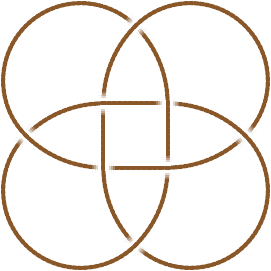} $8_{18}$   \\[1em]
	   \includegraphics [scale=.40,page=2] {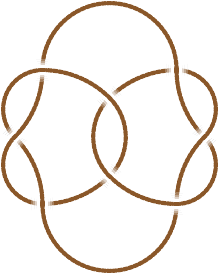} $8_{19}$  
	  & \includegraphics [scale=.40,page=2] {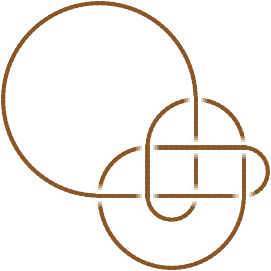} $8_{20}$  
	  & \includegraphics [scale=.40,page=2] {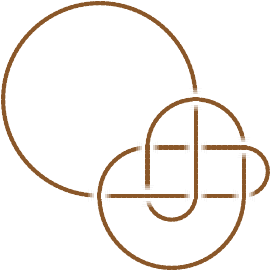} $8_{21}$  
	\end {tabular}
		\caption{Plane Lombardi drawings of all prime Lombardi knots up to~8 vertices.}
		\label{fig:primes}
	\end{figure}

\end{document}